\documentclass[11pt,english]{article}
\usepackage[T1]{fontenc}
\usepackage{textcomp}
\usepackage[latin9]{inputenc}
\usepackage[a4paper]{geometry}
\usepackage{parskip}
\usepackage{xcolor}
\usepackage{babel}
\usepackage{float}
\usepackage{units}
\usepackage{enumitem}
\usepackage{amsmath}
\usepackage{amsthm}
\usepackage{amssymb}
\usepackage{stackrel}
\usepackage{graphicx}
\usepackage[all]{xy}
\PassOptionsToPackage{normalem}{ulem}
\usepackage{ulem}
\usepackage[pdfusetitle,
 bookmarks=true,bookmarksnumbered=false,bookmarksopen=false,
 breaklinks=false,pdfborder={0 0 0},pdfborderstyle={},backref=false,colorlinks=true]
 {hyperref}
\hypersetup{
 linkcolor=blue,  citecolor=blue, urlcolor=blue}

\makeatletter

\newcommand{\noun}[1]{\textsc{#1}}
\providecommand{\tabularnewline}{\\}
\providecolor{lyxadded}{rgb}{0,0,1}
\providecolor{lyxdeleted}{rgb}{1,0,0}
\DeclareRobustCommand{\mklyxadded}[1]{\bgroup\color{lyxadded}{}#1\egroup}
\DeclareRobustCommand{\mklyxdeleted}[1]{\bgroup\color{lyxdeleted}\mklyxsout{#1}\egroup}
\DeclareRobustCommand{\mklyxsout}[1]{\ifx\\#1\else\sout{#1}\fi}

\numberwithin{equation}{section}
\numberwithin{figure}{section}
\numberwithin{table}{section}
\theoremstyle{plain}
\newtheorem{thm}{\protect\theoremname}[section]
\theoremstyle{definition}
\newtheorem{defn}[thm]{\protect\definitionname}
\theoremstyle{plain}
\newtheorem{question}[thm]{\protect\questionname}
\theoremstyle{plain}
\newtheorem*{thm*}{\protect\theoremname}
\theoremstyle{remark}
\newtheorem*{acknowledgement*}{\protect\acknowledgementname}
\theoremstyle{plain}
\newtheorem{prop}[thm]{\protect\propositionname}
\theoremstyle{plain}
\newtheorem{lem}[thm]{\protect\lemmaname}
\theoremstyle{plain}
\newtheorem{cor}[thm]{\protect\corollaryname}
\theoremstyle{remark}
\newtheorem{rem}[thm]{\protect\remarkname}

\usepackage{bbm}
\usepackage[abbrev]{amsrefs}
\usepackage{xypic}
\usepackage{stackrel,amssymb,amsmath}

\newcommand{\rleftarrows}{\mathrel{\raise.75ex\hbox{\oalign{%
  $\hfil\scriptstyle\relbar$\cr
  \vrule width0pt height.5ex$\scriptstyle\smash\leftarrow$\cr}}}}
\newcommand{\rightlarrows}{\mathrel{\raise.75ex\hbox{\oalign{%
  $\scriptstyle\rightarrow$\hfil\cr
  $\scriptstyle\vrule width0pt height.5ex\relbar$\cr}}}}
\newcommand{\Rrelbar}{\mathrel{\raise.75ex\hbox{\oalign{%
  $\scriptstyle\relbar$\cr
  \vrule width0pt height.5ex$\scriptstyle\relbar$}}}}
\newcommand{\longrightleftarrows}{\rleftarrows\joinrel\Rrelbar\joinrel\rightlarrows}

\makeatletter
\def\rightleftarrowsfill@{\arrowfill@\rleftarrows\Rrelbar\rightlarrows}
\newcommand{\xrightleftarrows}[2][]{\ext@arrow 3399\rightleftarrowsfill@{#1}{#2}}
\makeatother

\setlist[enumerate]{itemsep=-0.5\parsep, topsep=-0.5\parsep}

\DeclareMathOperator{\dist}{dist}
\DeclareMathOperator{\diag}{diag}

\DeclareMathOperator{\Stab}{Stab}

\DeclareMathOperator{\Spec}{Spec}
\DeclareMathOperator{\ord}{ord}

\DeclareMathOperator{\rad}{radius}

\renewcommand{\L}{\boldsymbol{L}}
\newcommand{\one}{\mathbbm{1}}

\makeatother

\providecommand{\acknowledgementname}{Acknowledgement}
\providecommand{\corollaryname}{Corollary}
\providecommand{\definitionname}{Definition}
\providecommand{\lemmaname}{Lemma}
\providecommand{\propositionname}{Proposition}
\providecommand{\questionname}{Question}
\providecommand{\remarkname}{Remark}
\providecommand{\theoremname}{Theorem}

\begin{document}
\title{Arithmeticity, thinness and efficiency\\
of qutrit Clifford+T gates}
\author{Shai Evra and Ori Parzanchevski}
\maketitle
\begin{abstract}
The Clifford+T gate set is a topological generating set for $PU(2)$,
which has been well-studied from the perspective of quantum computation
on a single qubit. The discovery that it generates a full S-arithmetic
subgroup of $PU(2)$ has led to a fruitful interaction between quantum
computation and number theory, resulting in a proof that words in
these gates cover $PU(2)$ in an almost-optimal manner.

In this paper we study the analogue gate set for $PU(3)$. We show
that in $PU(3)$ the group generated by the Clifford+T gates is not
arithmetic -- in fact, it is a \emph{thin }matrix group, namely a
Zariski-dense group of infinite index in its ambient S-arithmetic
group. 

On the other hand, we study a recently proposed extension of the Clifford+T
gates, called Clifford+D, and show that these do generate a full S-arithmetic
subgroup of $PU(3)$, and satisfy a slightly weaker almost-optimal
covering property than that of Clifford+T in $PU(2)$. The proofs
are different from those for $PU(2)$: while both gate sets act naturally
on a (Bruhat-Tits) tree, in $PU(2)$ the generated group acts transitively
on the vertices of the tree, and this is a main ingredient in proving
both arithmeticity and efficiency. In the $PU(3)$ Clifford+D case
the action on the tree is far from being transitive. This makes the
proof of arithmeticity considerably harder, and the study of efficiency
by automorphic representation theory becomes more involved, and results
in a covering rate which differs from the optimal one by a factor
of $\log_{3}(105)\approx4.236$.
\end{abstract}

\section{Introduction}

In quantum computation one is interested in approximating arbitrary
elements of the projective unitary group $PU(d)$ using ``circuits''
built from a fixed finite set of basic ``gates'' $\Sigma\subset PU(d)$.
To achieve this with arbitrary precision, the generated group $\left\langle \Sigma\right\rangle $
must be dense in $PU(d)$, and such sets are called universal. One
can further ask about efficiency, namely the rate at which words of
growing lengths in $\Sigma$ cover $PU(d)$ up to a desired error
term. 

A popular choice of gates in $PU(2)$ is the Clifford+T (or C+T, for
short) gate set: 
\[
C=\left\langle \left(\begin{matrix}1\\
 & i
\end{matrix}\right),\frac{1}{\sqrt{2}}\left(\begin{matrix}1 & 1\\
1 & -1
\end{matrix}\right)\right\rangle \qquad\text{and }\qquad T=\left(\begin{matrix}1\\
 & \zeta_{8}
\end{matrix}\right),
\]
where $\zeta_{n}$ denotes a primitive $n$-th root of unity. These
were shown to be universal in \cite{Boykin1999universalfaulttolerant},
but in \cite{kliuchnikov2013fast} it was furthermore shown that they
generate an S-arithmetic group, comprising all unitary matrices with
entries in $\mathbb{Z}[\zeta_{8},\frac{1}{2}]$. This allowed \cite{sarnak2015letter,Parzanchevski2018SuperGoldenGates}
to bring in deep number theoretic machinery (which goes back to \cite{lubotzky1986hecke,lubotzky1987hecke}),
and prove that the covering rate of $PU(2)$ by these gates is almost
optimal, in a precise sense (see Definition \ref{def:aoac}). 

The C+T gates were later generalized to higher $PU(d)$, whenever
$d$ is a prime \cite{Howard2012Quditversionsqubit}, but little is
known on the arithmetic properties of these generalizations. In this
paper we resolve the case of $d=3$, for the generalized C+T gates,
as well as an extension of these, called C+D, which was suggested
in \cite{Kalra2023SynthesisAirthmeticSingle} using the notion of
higher Clifford hierarchy \cite{Cui2017DiagonalgatesClifford}. Both
of these gate sets turn out to be highly interesting from the mathematical
perspective, and present several new phenomena which we discuss below.
\begin{defn}
\label{def:C+T-and-C+D}The C+T and C+D gate sets in $PU(3)$ are
\[
C\!+\!T=\left\{ H,S,T\right\} ,\qquad\mbox{and}\qquad C\!+\!D=\left\{ H\right\} \cup\mathcal{D},
\]
where 
\begin{align*}
H= & \frac{1}{\sqrt{-3}}\left(\begin{array}{ccc}
1 & 1 & 1\\
1 & \zeta_{3} & \overline{\zeta_{3}}\\
1 & \overline{\zeta_{3}} & \zeta_{3}
\end{array}\right), & S= & \left(\begin{matrix}1\\
 & \zeta_{3}\\
 &  & 1
\end{matrix}\right),\\
T= & \left(\begin{matrix}\zeta_{9}\\
 & 1\\
 &  & \zeta_{9}^{-1}
\end{matrix}\right), & \mathcal{D}= & \left\{ \left(\begin{array}{ccc}
\pm\zeta_{9}^{a}\\
 & \pm\zeta_{9}^{b}\\
 &  & \pm\zeta_{9}^{b}
\end{array}\right)\,\middle|\,a,b,c\in\mathbb{Z}/9\mathbb{Z}\right\} .
\end{align*}
We denote by $\Gamma$ the group of all matrices in $U(3)$ with entries
in $\mathbb{Z}\left[\zeta_{9},\tfrac{1}{3}\right]$, and note that
$H,S,T\in\Gamma$ and $\mathcal{D}\leq\Gamma$. It was shown in \cite{Glaudell2019Canonicalformssingle}
that the C+T gates do not generate $\Gamma$, which prompted \cite{Kalra2023SynthesisAirthmeticSingle}
to suggest the C+D extension and raise the following questions:
\end{defn}

\begin{question}[\cite{Kalra2023SynthesisAirthmeticSingle}]
\label{que:main}
\begin{enumerate}
\item Does the C+D gate set generate $\Gamma$?
\item If so, does there exist an efficient algorithm to write a matrix in
$\Gamma$ using C+D?
\end{enumerate}
\end{question}

In Section \ref{sec:Synthesis-and-arithmeticity} of this paper we
answer both parts of Question \ref{que:main} in the affirmative:
\begin{thm*}[\ref{thm:main}]
The answer to both parts of Question \ref{que:main} is yes.
\end{thm*}
Similar results for other groups were achieved in the past: the C+T
case for $PU(2)$ in \cite{kliuchnikov2013fast}, and others in \cite{lubotzky1986hecke,sarnak2015letter,kliuchnikov2015framework,Forest2015Exactsynthesissingle,Parzanchevski2018SuperGoldenGates,Evra2018RamanujancomplexesGolden,Evra2022Ramanujanbigraphs}.
All of these cases share a very special feature: the group generated
by the gates acts transitively on the vertices of a Bruhat-Tits tree
or building (or on the set of all vertices of a fixed color), and
the gate set takes some fixed vertex to all of its closest neighbors.
In this paper this turns out to be far from the case! In fact, Corollary
\ref{cor:dist>=00003D6} shows that $\Gamma$ acts on its Bruhat-Tits
tree with a rather sparse orbit, so that a different approach is needed.
What we are able to show in Section \ref{sec:Synthesis-and-arithmeticity},
using a mixture of number theory, combinatorics, and some good fortune
(e.g.\ in (\ref{eq:reduct})) is that the $\Gamma$-orbit of a specific
vertex in the tree can also be covered by words in C+D (Theorem \ref{thm:orbits}),
and this suffices for the purpose of proving Theorem \ref{thm:main}.
\vspace{.5\baselineskip}

Part 1 of Question \ref{que:main} is a special case of the following
type of problems: Does a given finite set of elements in an arithmetic
or S-arithmetic group generate a finite-index subgroup? To quote \cite{sarnak2014notes},
such problems can be formidable, and generically the answer to them
is no, though specific cases may be hard to prove (see the book \cite{breuillard2014thin},
and especially the chapters by Fuchs and Sarnak). A subgroup $\Delta$
of an S-arithmetic matrix group $\Gamma$ is called \emph{thin }if
it is of infinite index in $\Gamma$, and at the same time Zariski
dense. The underlying number field plays a role in the interpretation
of Zariski denseness: in our case, since $\mathbb{Q}\left(\zeta_{9}\right)$
has three non-conjugate embeddings in $\mathbb{C}$, the group $\Gamma$
naturally embeds in $PU(3)^{3}=PU(3)\times PU(3)\times PU(3)$, using
all three together. Considering $\Gamma$ as a $\mathbb{Q}\left(\cos(\tfrac{2\pi}{9})\right)$-arithmetic
group, Zariski denseness is only equivalent to being dense in each
component separately. However, considering $\Gamma$ as a $\mathbb{Q}$-arithmetic
group\footnote{In algebro-geometric terms, this is the Weil restriction $\mathrm{Res}_{\mathbb{Q}(\cos(2\pi/9))/\mathbb{Q}}PU_{3}$.},
Zariski denseness is equivalent to being dense in $PU(3)^{3}$. Our
second main theorem is that the C+T gate set generates a thin matrix
group, in the strong sense:
\begin{thm*}[\ref{thm:C+T-thin}]
The C+T gates generate a thin matrix group in $\Gamma\leq PU(3)^{3}$.
\end{thm*}
In terms of quantum computation, this implies that C+T is even universal
in $PU(3)^{3}$, namely, for any three gates in $PU(3)$ and $\varepsilon>0$,
there is a single circuit in the C+T gates whose three complex embeddings
$\varepsilon$-approximate the three original gates simultaneously.
In order to prove this theorem, we develop in Section \ref{sec:Amalgamation-thin}
general criteria to show that a subgroup of an S-arithmetic group
in $PU(3)$ is of infinite index (Proposition \ref{prop:B-S}), and
Zariski-dense (Proposition \ref{prop:Weigel}). For the former we
make use of Bass-Serre theory, and for the latter we use Weigel's
study of Frattini extensions in $p$-adic integer groups. In Section
\ref{subsec:-Thinness-of-Clifford+T} we specialize to the C+T gate
set and prove its thinness using our criteria, and along the way we
obtain more results on the structure of the group $\Gamma$, such
as a neat presentation as an amalgamated product (Corollary \ref{cor:Bass-Serre}).

\vspace{.5\baselineskip}

In Section \ref{sec:Covering-rate} we turn to study the covering
rate of $PU(3)$ (or $PU(3)^{d}$) by various families of subsets.
For a set $\Sigma\subset PU(d)$ we denote by $\Sigma^{r}$ the words
of length $r$ in $\Sigma$, and define the growth rate of $\Sigma$
by $\rho=\rho(\Sigma)=\lim_{r\rightarrow\infty}\sqrt[r]{|\Sigma^{r}|}$.
We say that $\gamma\in\Sigma^{r}$ is an \emph{$\varepsilon$-approximation}
of $g\in PU(3)$ if $g$ lies in the ball of \emph{volume} $\varepsilon$
around $\gamma$ in $PU(3)$. We normalize the volume of $PU(3)$
to be one, so that it is clear that $\left|\Sigma^{r}\right|\geq\frac{1}{\varepsilon}$
is needed to $\varepsilon$-approximate all of $PU(3)$. By the Solovay-Kitaev
Theorem, there exist $c\geq1$, $K>0$, such that for any symmetric
universal gate set $\Sigma\subset PU(3)$ and any small enough $\varepsilon>0$,
all of $PU(3)$ is $\varepsilon$-approximated by $\Sigma^{r}$ for
$r\leq K\big(\log_{\rho}\frac{1}{\varepsilon}\big)^{c}$. The current
best known bound on the exponent $c$ for a general universal $\Sigma$
is $c\sim1.44$ \cite{Kuperberg2023Breakingcubicbarrier}. However,
the work of Bourgain and Gamburd \cite{bourgain2012spectral} shows
that whenever $\Sigma$ is algebraic (i.e.\ all the entries of its
elements are in $\overline{\mathbb{Q}}$), the associated averaging
operator on $PU(3)$ has a spectral gap, which by \cite{Harrow2002Efficientdiscreteapproximations}
implies that they achieve the optimal exponent $c=1$. Thus, their
covering efficiency is measured by the constant $K$ (which is very
close to the ``covering exponent'' of $\Sigma$ defined in \cite{sarnak2015letter}).
Section \ref{sec:Covering-rate} ultimately leads to a proof that
the C+D gates achieve Solovay-Kitaev with $c=1$ and any $K>\log_{3}105\approx4.236$:
\begin{thm*}[\ref{thm:H-covering}]
For any $K>\log_{3}(105)$, for any small enough $\varepsilon$ every
$g\in PU(3)^{3}$ has an $\varepsilon$-approximation by a word in
the C+D gate set of length at most $K\log_{\rho}\tfrac{1}{\varepsilon}$.
\end{thm*}
A stronger notion of covering efficiency is that of \emph{Golden Gates},
which was developed in \cite{sarnak2015letter,Parzanchevski2018SuperGoldenGates,Evra2018RamanujancomplexesGolden}.
These gates exhibit an almost-optimal almost-covering (a.o.a.c.) property
(see Definition \ref{def:aoac}), which implies that if there are
$\Theta(\rho^{r})$ words of length $r$, then \emph{almost} every
$g$ has an $\varepsilon$-approximation of the almost-optimal length
$r=\log_{\rho}\frac{1}{\varepsilon}+O\big(\log\log\frac{1}{\varepsilon}\big)$.
It also implies that \emph{every} $g$ has an $\varepsilon$-approximation
of length $2\log_{\rho}\frac{1}{\varepsilon}+O\big(\log\log\frac{1}{\varepsilon}\big)$,
so that Solovay-Kitaev holds for $c=1$ any $K>2$. 

Golden gate sets for $PU(3)$ were constructed in \cite{Evra2018RamanujancomplexesGolden,Evra2022Ramanujanbigraphs};
in these papers special universal sets were carefully chosen to ensure
the a.o.a.c.\ property by number-theoretic arguments. In this paper
we face a problem of a different nature: the gates are given to us,
and the techniques developed in \cite{Evra2018RamanujancomplexesGolden,Evra2022Ramanujanbigraphs}
do not apply here, mostly due to the failure of $\Gamma$ to act transitively
on the Bruhat-Tits tree. 

Rather than working on the Clifford gates directly, we study in Section
\ref{sec:Covering-rate} the covering rate of a general $\Pi$-arithmetic
group $\Gamma$ in $PU(3)$, under a mild technical restriction (having
Iwahori level at a ramified prime). Using the action of $\Gamma$
on the $\Pi$-adic Bruhat-Tits tree we introduce two families of subsets
in $\Gamma$, which we call Clozel-Hecke points (Definition \ref{def:Clozel-Hecke})
and level points ((\ref{eq:level-sets})). In Corollaries \ref{cor:Clozel-a.o.a.c.}
and \ref{cor:level-a.o.a.c.} we prove that both of these families
exhibit the a.o.a.c.\ property. To achieve this we translate the
covering problem to one on automorphic representations of adelic groups,
and employ a Ramanujan-type result from \cite{Evra2022Ramanujanbigraphs},
which itself build upon deep number theoretic results pertaining to
the Langlands program \cite{Shin2011Galoisrepresentationsarising}. 

In Section \ref{sec:Silver-gates} we specialize again to the Clifford
gates, using coarse geometry to translate the covering results for
the Clozel-Hecke and level points to ones on words in the C+D gates.
For this we combine the theory of Bass-Serre normal form with the
explicit synthesis/word problem algorithm from Section \ref{sec:Synthesis-and-arithmeticity}.
In addition to Theorem \ref{thm:H-covering} which was quoted above,
we obtain from this analysis that almost every $g$ has an $\varepsilon$-approximation
of length at most $3\log_{\rho}\frac{1}{\varepsilon}+O\big(\log\log\frac{1}{\varepsilon}\big)$
(Corollary \ref{cor:H-a.o.a.c.}). This is a slightly weaker almost-covering
property from that of Golden Gates, earning Clifford+D the title of
a ``silver'' gate set.
\begin{acknowledgement*}
We thank A.\ R.\ Kalra for presenting this problem to us. S.E.\ was
supported by ISF grant 1577/23, and O.P.\ was supported by ISF grant
2990/21.
\end{acknowledgement*}

\section{Synthesis and arithmeticity for Clifford+D\protect\label{sec:Synthesis-and-arithmeticity}}

The goal of this section is to prove the arithmeticity of the C+D
gate set. We begin by describing compact unitary groups in three variables
in general, but after Proposition \ref{prop:level_dist} we restrict
to the case of the Clifford gates, for which the definitions simplify
considerable.

Let $F$ be a totally real number field, $E/F$ a CM extension, $\mathcal{O}=\mathcal{O}_{F}$
and $\mathcal{O}_{E}$ the rings of integers of $F$ and $E$, $\varepsilon_{1},\ldots,\varepsilon_{d}\colon F\hookrightarrow\mathbb{R}$
the real embeddings ($d=[F:\mathbb{Q}]$), and $\Phi\in GL_{3}\left(E\right)$
a totally definite Hermitian form.
\begin{defn}
\label{def:U3-z9-1}The \emph{unitary group scheme} $G=U_{3}^{E,\Phi}$
associated with the form $\Phi$ is defined by assigning to every
$\mathcal{O}$-algebra $A$ the group
\[
G\left(A\right)=U_{3}\left(A\otimes_{\mathcal{O}}\mathcal{O}_{E}\right)=\left\{ g\in GL_{3}\left(A\otimes_{\mathcal{O}}\mathcal{O}_{E}\right)\,\middle|\,g^{*}\Phi g=\Phi\right\} .
\]
We shall also consider on occasions the \emph{special unitary group}
$G'=SU_{3}^{E,\Phi}$ (defined by adding $\det g=1$) and the \emph{projective
group schemes} $\overline{G}=PU_{3}^{E,\Phi}=G/Z(G)$ ($G$ modulo
its center), and $\overline{G}'=PSU_{3}^{E,\Phi}=G'/Z(G')$. 
\end{defn}

Since $E/F$ is a CM-extension and $\Phi$ is totally definite, the
group of real points of $G$ is
\begin{equation}
G\left(F\otimes_{\mathbb{Z}}\mathbb{R}\right)=G(F_{\varepsilon_{1}}\times\ldots\times F_{\varepsilon_{d}})=G(F_{\varepsilon_{1}})\times\ldots\times G(F_{\varepsilon_{d}})\cong U(3)^{d}.\label{eq:real-points}
\end{equation}
For a prime ideal $\Pi$ in $\mathcal{O}$\textcolor{red}{{} }which
does not split over $E$ we consider the $\Pi$-integers 
\begin{align*}
\mathcal{O}\big[\smash{\tfrac{1}{\Pi}}\big] & =\left\{ \alpha\in F\,\middle|\,v(\alpha)\geq0\text{ for every finite valuation of \ensuremath{F} other than \ensuremath{v_{\Pi}}}\right\} \\
 & =\left\{ \smash{\tfrac{\alpha}{\beta}}\,\middle|\,\alpha,\beta\in\mathcal{O},\beta\notin\Pi'\text{ for every prime ideal }\Pi'\neq\Pi\right\} ,
\end{align*}
and study the $\Pi$-arithmetic group $\Gamma=G\big(\mathcal{O}\big[\tfrac{1}{\Pi}\big]\big)$,
which is naturally embedded in $U(3)^{d}$ via (\ref{eq:real-points}).
The group $\Gamma$ acts naturally on an infinite tree $\mathcal{T}$,
which is the Bruhat-Tits building associated with the $\Pi$-adic
group $G_{\Pi}:=G(F_{\Pi})$ (where $F_{\Pi}$ is the $\Pi$-adic
completion of $F$). It is simplest to describe $\mathcal{T}$ using
the (reduced) Bruhat-Tits building $\mathcal{B}$ of the group $\widetilde{G}:=GL_{3}\left(E_{\pi}\right)$,
where $\pi$ is a prime factor of $\Pi$ in $\mathcal{O}_{E}$. This
is a 2-dimensional building, whose vertices correspond to the cosets
$\widetilde{G}/\widetilde{K}$, where $\widetilde{K}:=E_{\pi}^{\times}GL_{3}(\mathcal{O}_{E_{\pi}})$
(here $\mathcal{O}_{E_{\pi}}$ are the $\pi$-adic integers in $E_{\pi}$).
Three vertices $\{g_{i}\widetilde{K}\}_{i=1,2,3}$ form a triangle
in $\mathcal{B}$ iff they give rise to a chain of $\mathcal{O}_{E_{\pi}}$-lattices
$\pi g_{3}\mathcal{O}_{E_{\pi}}^{3}<g_{1}\mathcal{O}_{E_{\pi}}^{3}<g_{2}\mathcal{O}_{E_{\pi}}^{3}<g_{3}\mathcal{O}_{E_{\pi}}^{3}$,
possibly after permuting the $g_{i}$, and scaling each by some power
of $\pi$ (see \cite[V(8)]{Brown1989} for more details). We shall
assume in this paper that $\Phi\in\widetilde{K}$, as this is a bit
simpler and covers the cases we are interested in (see \cite[§5.1]{Evra2022Ramanujanbigraphs}
for the general case). 

The group $G_{\Pi}$ consists of the fixed-points of the involution
$g^{\#}=\Phi^{-1}(g^{*})^{-1}\Phi$ on $\widetilde{G}$, and $\#$
induces a simplicial involution on $\mathcal{B}$, via $(g\widetilde{K})^{\#}=g^{\#}\widetilde{K}$.
The tree $\mathcal{T}$ is the fixed-section of this involution, and
it has a bipartite decomposition $\mathrm{Ver}_{\mathcal{T}}=L_{\mathcal{T}}\sqcup R_{\mathcal{T}}$,
where $L_{\mathcal{T}}$ consists of the $\mathcal{B}$-vertices fixed
by $\#$, and $R_{\mathcal{T}}$ of the midpoints of $\mathcal{B}$-edges
which $\#$ reflects. If $\Pi$ is ramified in $E$ then $\mathcal{T}$
is a regular tree of degree $N_{F/\mathbb{Q}}(\Pi)+1$, and if $\Pi$
is inert then $\mathcal{T}$ is bi-regular with $L_{\mathcal{T}}$
having degrees $N_{F/\mathbb{Q}}(\Pi)^{3}+1$ and $R_{\mathcal{T}}$
degrees $N_{F/\mathbb{Q}}(\Pi)+1$. The group $G_{\Pi}$ acts transitively
on the edges of $\mathcal{T}$, which implies that it acts transitively
on $L_{\mathcal{T}}$ and $R_{\mathcal{T}}$. We denote by $v_{0}$
the vertex corresponding to $\widetilde{K}$ itself, which is in $L_{\mathcal{T}}$
and has stabilizer $K_{\Pi}:=G_{\Pi}\cap\widetilde{K}=G(\mathcal{O}_{F_{\Pi}})$.

Let $\mathrm{ord}_{\pi}\colon E_{\pi}\rightarrow\mathbb{Z}$ be the
$\pi$-valuation on $E_{\pi}$, normalized by $\mathrm{ord}_{\pi}\left(\pi\right)=1$.
For $g\in G_{\Pi}$ denote
\[
\mathrm{ord}_{\pi}g=\min\nolimits_{i,j}\ord_{\pi}(g_{ij})=\max\left\{ t\,\,\middle|\,\,\pi^{-t}g\in M_{3}\left(\mathcal{O}_{E_{\pi}}\right)\right\} ,
\]
and define the \emph{level }map on $G_{\Pi}$ to be 
\begin{equation}
\ell\colon G_{\Pi}\rightarrow2\mathbb{N},\qquad\ell\left(g\right)=-2\mathrm{ord}_{\pi}g.\label{eq:l}
\end{equation}

\begin{prop}
\label{prop:level_dist}The distance in $\mathcal{T}$ between vertices
in $L_{\mathcal{T}}$ is given by $\mathrm{dist}\left(gv_{0},hv_{0}\right)=\ell\left(h^{-1}g\right)$,
for any $g,h\in G_{\Pi}$.
\end{prop}

\begin{proof}
We have $\mathrm{dist}\left(gv_{0},hv_{0}\right)=\mathrm{dist}\left(h^{-1}gv_{0},v_{0}\right)$,
and it is proved in Proposition 3.3 of \cite{Evra2018RamanujancomplexesGolden}
that the latter equals $\ell\left(h^{-1}g\right)$ (the claim there
is for $E=\mathbb{Q}\left[i\right]$ and $\pi\in\mathbb{Z}[i]$ an
unramified prime, but the proof holds more generally to our case).
\end{proof}
\medskip{}

For the rest of this section we restrict to the specific case which
corresponds to the extended Clifford gates, for which the story is
simpler than the general case. We denote $\xi=\zeta_{9}$ and $\sigma=\xi+\xi^{-1}$,
and take $E=\mathbb{Q}\left[\xi\right]$ to be the $9$-th cyclotomic
field and $F=\mathbb{Q}\left[\sigma\right]$ its maximal totally
real subfield. The rings of integers of $E$ and $F$ are $\mathcal{O}_{E}=\mathbb{Z}\left[\xi\right]$
and $\mathcal{O}_{F}=\mathbb{Z}\left[\sigma\right]$, and both of
them are PID. We take $\Phi=I$ to be the standard Hermitian form,
so that $G=U_{3}^{E,\Phi}$ is given by 
\[
G\left(A\right)=U_{3}\left(A\left[\xi\right]\right)=\left\{ g\in GL_{3}\left(A\left[\xi\right]\right)\,\middle|\,g^{*}g=I\right\} ,
\]
where $A[\xi]=\nicefrac{A[x]}{(m_{\xi}^{F}(x))}$, for $m_{\xi}^{F}(x)=x^{2}-\sigma x+1$
the minimal polynomial of $\xi$ over $F$.

We denote $\pi=1-\xi$ and take $\Pi=\pi\bar{\pi}=2-\sigma\in\mathcal{O}_{F}$.
We note that $\pi$ (resp. $\Pi$) is a prime in $\mathcal{O}_{E}$
(resp. $\mathcal{O}_{F}$) with $\pi^{6}\sim3$ (resp. $\Pi^{3}\sim3$),
and observe that the $\Pi$-arithmetic group $\Gamma$ is
\[
\Gamma=G\big(\mathcal{O}_{F}\big[\tfrac{1}{\Pi}\big]\big)=\left\{ g\in U(3)\,\middle|\,\text{all the entries of \ensuremath{g} are in \ensuremath{\mathbb{Z}\left[\zeta_{9},\tfrac{1}{3}\right]}}\right\} ,
\]
which indeed contains all the gates from Definition \ref{def:C+T-and-C+D}.
Let $\mathrm{Gal}\left(E/\mathbb{Q}\right)=\left\langle \varphi\right\rangle \cong\nicefrac{\mathbb{Z}}{6}$
with $\varphi\left(\xi\right)=\xi^{2}$. Since $\varphi^{3}$ generates
the Galois group of the CM-extension $E/F$, we shall denote $\varphi^{3}\left(\alpha\right)$
by $\overline{\alpha}$. We take $\varepsilon_{1}$ to be the real
embedding $\varepsilon_{1}\colon\sigma\mapsto2\cos\left(\tfrac{2\pi}{9}\right)\colon F\hookrightarrow\mathbb{R}$,
and let $\varepsilon_{2}=\varepsilon_{1}\circ\varphi$ and $\varepsilon_{3}=\varepsilon_{1}\circ\varphi^{2}$
be the two other real embeddings of $F$. They relate to the (absolute)
norm of $E$ by
\begin{equation}
\prod\nolimits_{i=1}^{3}\varepsilon_{i}(\overline{\alpha}\alpha)=\varepsilon_{1}(\overline{\alpha}\alpha\varphi(\overline{\alpha}\alpha)\varphi^{2}(\overline{\alpha}\alpha))=\varepsilon_{1}(N_{E/\mathbb{Q}}(\alpha))=N_{E/\mathbb{Q}}(\alpha)\qquad\left(\forall\alpha\in E\right).\label{eq:norm_emb}
\end{equation}
As mentioned, both $\mathcal{O}_{E}$ and $\mathcal{O}_{F}$ are
PID (and in particular UFD), with unit groups\footnote{These computations were carried out in sage \cite{SageDevelopers2023SageMathSageMathematics},
which itself relies on PARI \cite{2022PARI/GPversion2.15.2}.}
\begin{align}
\mathcal{O}_{E}^{\times} & =\left\langle -\xi\right\rangle \times\left\langle u_{1}:=1+\xi\right\rangle \times\left\langle u_{2}:=1+\xi^{2}\right\rangle \cong\nicefrac{\mathbb{Z}}{18}\times\mathbb{Z}\times\mathbb{Z}\nonumber \\
\mathcal{O}_{F}^{\times} & =\left\langle -1\right\rangle \times\left\langle 1-\sigma\right\rangle \times\left\langle \sigma\right\rangle \cong\nicefrac{\mathbb{Z}}{2}\times\mathbb{Z}\times\mathbb{Z}.\label{eq:OFx}
\end{align}
We shall also need the unitary subgroup $U_{E/F}^{1}=\left\{ \alpha\in\mathcal{O}_{E}^{\times}\,\middle|\,N_{E/F}(\alpha)=\overline{\alpha}\alpha=1\right\} $,
and since $\overline{u_{1}}u_{1}=(1-\sigma)^{-2}$ and $\overline{u_{2}}u_{2}=\sigma^{2}$,
we see from (\ref{eq:OFx}) that $U_{E/F}^{1}=\left\langle -\xi\right\rangle $.

Since $\Pi$ is ramified in $E$ and $N_{F/\mathbb{Q}}(\Pi)=3$, the
Bruhat-Tits tree of $G_{\Pi}$ is $4$-regular. By Proposition \ref{prop:level_dist},
the level-zero elements in $G_{\Pi}$ are precisely $\mathrm{Stab}_{G_{\Pi}}(v_{0})=K_{\Pi}$.
We denote:
\[
C:=\mathrm{Stab}_{\Gamma}\left(v_{0}\right)=\Gamma\cap K_{\Pi}=G\left(\mathcal{O}_{F}\big[\tfrac{1}{\Pi}\big]\right)\cap G\left(\mathcal{O}_{F_{\Pi}}\right)=G\left(\mathcal{O}_{F}\right).
\]

\begin{lem}
\label{lem:C}
\begin{enumerate}
\item The group $C$ consists of the monomial matrices with entries in $\left\langle -\xi\right\rangle $:
\begin{equation}
C=\mathcal{M}_{3}\ltimes\mathcal{D}\cong S_{3}\ltimes\left(\nicefrac{\mathbb{Z}}{18}\right)^{3},\label{eq:C-M3-D}
\end{equation}
where $\mathcal{M}_{3}$ are the permutation matrices.
\item The group $\langle H,\mathcal{D}\rangle$ contains $C$.
\end{enumerate}
\end{lem}

\begin{proof}
\emph{(1)} Let $g\in C$. Since $g^{*}g=I$ we have $\sum_{\ell=1}^{3}\overline{g_{k\ell}}g_{k\ell}=1$
for any fixed $k$ (and similarly when we swap the roles of $k$ and
$\ell$). Note that $\varepsilon_{i}\left(\bar{\alpha}\alpha\right)>0$
for any $\alpha\in E^{\times}$ and any $i$, and that if $0\neq\alpha\in\mathcal{O}_{E}$
and $\varepsilon_{i}\left(\bar{\alpha}\alpha\right)<1$ then there
exists $j\ne i$ such that $\varepsilon_{j}\left(\bar{\alpha}\alpha\right)>1$
by (\ref{eq:norm_emb}), since $N_{E/\mathbb{Q}}(\alpha)\in\mathbb{Z}$.
Hence for any $k$ there exists a unique $\ell$ such that $\overline{g_{k\ell}}g_{k\ell}=1$
and $\overline{g_{m\ell}}g_{m\ell}=0$ for $m\ne\ell$, namely, $g$
is a monomial matrix with coefficients in $U_{E/F}^{1}=\left\langle -\xi\right\rangle $.

\emph{(2)} This follows from the fact that for $W=\diag(1,\xi^{3},\xi^{6})$
we have 
\begin{equation}
\mathcal{M}_{3}=\left\{ 1,-H^{2},-HWH,-HW^{2}H,H^{3}WH,H^{3}W^{2}H\right\} .\qedhere\label{eq:M3-fromCD}
\end{equation}
\end{proof}
\begin{prop}
\label{prop:l>5} If $\gamma\in\Gamma$ and $\ell\left(\gamma\right)<6$
then $\gamma\in C$.
\end{prop}

\begin{proof}
Denote $\mathfrak{e}=\ell\left(\gamma\right)/2=-\ord_{\pi}\gamma$,
and let $g=\pi^{\mathfrak{e}}\gamma$, which is in $M_{3}(\mathcal{O}_{E})$
and satisfies $g^{*}g=\Pi^{\mathfrak{e}}I$. If $\mathfrak{e}=0$,
this means that $g^{*}g=I$ and thus $\gamma=g\in C$, so we can assume
from now on $\mathfrak{e}\in\left\{ 1,2\right\} $. Denoting the first
row of $g$ by $\left(\alpha,\beta,\gamma\right)$, we observe that
$\varepsilon_{i}(\overline{\alpha}\alpha)+\varepsilon_{i}(\overline{\beta}\beta)+\varepsilon_{i}(\overline{\gamma}\gamma)=\varepsilon_{i}(\Pi)^{\mathfrak{e}}$
for $1\leq i\leq3$, which forces $\varepsilon_{i}(\overline{\alpha}\alpha)\leq\varepsilon_{i}(\Pi)^{\mathfrak{e}}$.
Using (\ref{eq:norm_emb}), we obtain
\begin{equation}
N_{E/\mathbb{Q}}(\alpha)=\prod\nolimits_{i=1}^{3}\varepsilon_{i}(\overline{\alpha}\alpha)\leq\prod\nolimits_{i=1}^{3}\varepsilon_{i}(\Pi)^{\mathfrak{e}}=\prod\nolimits_{i=1}^{3}\varepsilon_{i}(\overline{\pi}\pi)^{\mathfrak{e}}=N_{E/\mathbb{Q}}(\pi)^{\mathfrak{e}}=3^{\mathfrak{e}}.\label{eq:norm_ineq}
\end{equation}
Assume now that $\alpha\neq0$, let $\mathfrak{p}\in\mathcal{O}_{E}$
be a prime factor of $\alpha$, and let $p\in\mathbb{Z}$ be the prime
below it. Since $N_{E/\mathbb{Q}}(\mathfrak{p})$ is a positive power
of $p$, $N_{E/\mathbb{Q}}(\mathfrak{p})\leq3^{\mathfrak{e}}\leq9$
forces $p\leq7$. Furthermore, for $p=2,3,5,7$ we have $N_{E/\mathbb{Q}}(\mathfrak{p})=2^{6},3,5^{6},7^{3}$
respectively (for any $\mathfrak{p}$ above $p$), so that we must
have $\mathfrak{p}=\pi$, up to associates. Thus, we can write $\alpha=(-\xi)^{r}u_{1}^{x}u_{2}^{y}\pi^{z}$,
and $3^{z}=N_{E/\mathbb{Q}}(\alpha)\leq3^{\mathfrak{e}}$ forces $z\leq\mathfrak{e}$.
From $\overline{\alpha}\alpha=(1-\sigma)^{-2x}\sigma^{2y}\Pi^{z}$
and $\overline{\alpha}\alpha+\overline{\beta}\beta+\overline{\gamma}\gamma=\Pi^{\mathfrak{e}}$
we obtain that 
\begin{equation}
\varepsilon_{i}((1-\sigma)^{-2})^{x}\varepsilon_{i}(\sigma^{2})^{y}\leq\varepsilon_{i}(\Pi^{\mathfrak{e}-z})\qquad(1\leq i\leq3),\label{eq:ineq}
\end{equation}
and (a 3-digit approximation of) the relevant real values is:\medskip{}
\\
\hspace*{\fill}%
\begin{tabular}{|c|c|c|c|c|}
\hline 
$\eta$ & $\overline{\eta}\eta$ & $\varepsilon_{1}(\overline{\eta}\eta)$ & $\varepsilon_{2}(\overline{\eta}\eta)$ & $\varepsilon_{3}(\overline{\eta}\eta)$\tabularnewline
\hline 
\hline 
$u_{1}=1+\xi$ & $(1-\sigma)^{-2}$ & 3.53 & 2.35 & 0.121\tabularnewline
\hline 
$u_{2}=1+\xi^{2}$ & $\sigma^{2}$ & 2.35 & 0.121 & 3.53\tabularnewline
\hline 
$\pi=1-\xi$ & $\Pi=2-\sigma$ & 0.468 & 1.65 & 3.88\tabularnewline
\hline 
$\pi^{2}=(1-\xi)^{2}$ & $\Pi^{2}=(2-\sigma)^{2}$ & 0.219 & 2.73 & 15.0\tabularnewline
\hline 
\end{tabular}\hspace*{\fill}

\medskip{}
For each value of $\mathfrak{e}-z$, we obtain from (\ref{eq:ineq})
(by taking log) three linear inequalities in $x$ and $y$, whose
common solutions (for $\mathfrak{e}-z\in\left\{ 1,2\right\} $) are
shown in Figure \ref{fig:ineq}. 
\begin{figure}[H]
\begin{minipage}[t]{0.49\columnwidth}%
\begin{center}
\includegraphics[width=0.75\columnwidth]{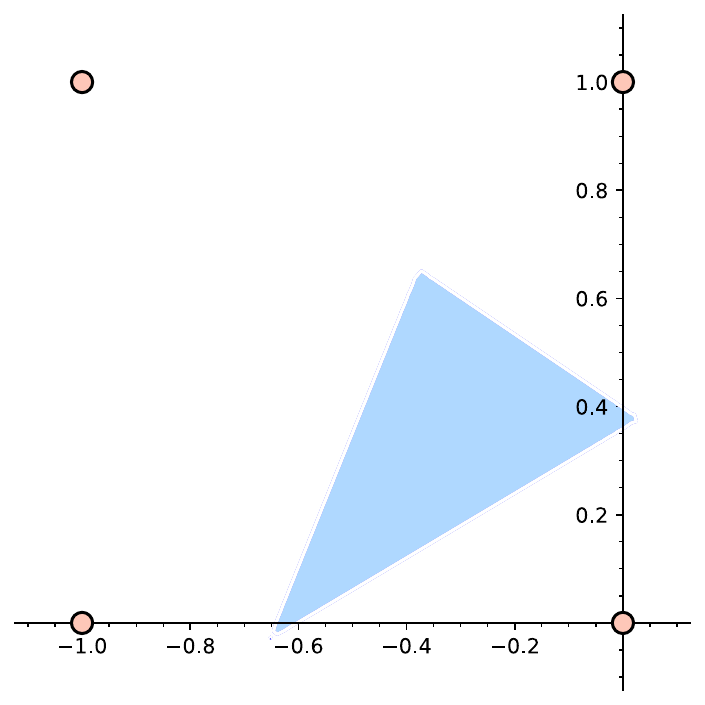}
\par\end{center}%
\end{minipage}%
\begin{minipage}[t]{0.49\columnwidth}%
\begin{center}
\includegraphics[width=0.75\columnwidth]{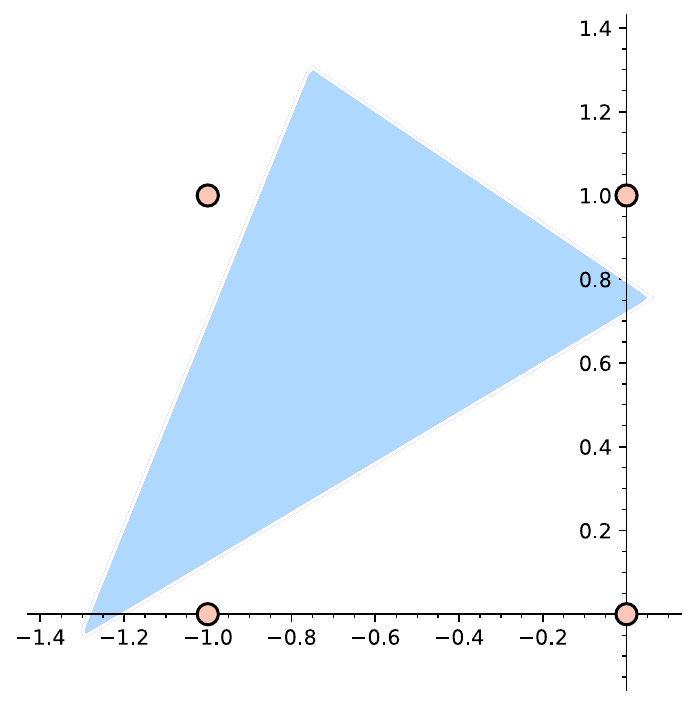}
\par\end{center}%
\end{minipage}

\caption{\protect\label{fig:ineq}The set of $\left(x,y\right)\in\mathbb{R}^{2}$
which satisfy the inequalities in (\ref{eq:ineq}) for $\mathfrak{\mathfrak{e}}-z=1$
(left) and $\mathfrak{\mathfrak{e}}-z=2$ (right). In both cases there
are no integral solutions.}
\end{figure}
For both $\mathfrak{e}-z\in\left\{ 1,2\right\} $ there is no integral
solution, so that we must have $\mathfrak{e}-z=0$. But then (\ref{eq:norm_ineq})
becomes an equality, which forces an equality in $\varepsilon_{i}(\overline{\alpha}\alpha)\leq\varepsilon_{i}(\Pi)^{\mathfrak{e}}$
for each $i$ separately, so from $\varepsilon_{1}(\overline{\alpha}\alpha)=\varepsilon_{1}(\Pi)^{\mathfrak{e}}$
we conlcude that $\beta=\gamma=0$. We obtained $(1-\sigma)^{-2x}\sigma^{2y}\Pi^{\mathfrak{e}}=\overline{\alpha}\alpha=\Pi^{\mathfrak{e}}$,
and from (\ref{eq:OFx}) we infer that $x=y=0$, i.e.\ $\alpha=\left(-\xi\right)^{r}\pi^{\mathfrak{e}}$.
We have assumed $\alpha\neq0$, but if $\alpha=0$ then the same analysis
holds for either $\beta$ or $\gamma$. The same goes for every row
and column in $g$, and in total we have obtained that $\gamma=\pi^{-\mathfrak{e}}g$
is monomial with entries in $\left\langle -\xi\right\rangle $.
\end{proof}
From Propositions \ref{prop:level_dist} and \ref{prop:l>5} we obtain:
\begin{cor}
\label{cor:dist>=00003D6}Let $v,w\in\Gamma v_{0}$. Then either $v=w$,
or $\mathrm{dist}\left(v,w\right)\geq6$.
\end{cor}

\begin{rem}
We note that \cite{Kalra2023SynthesisAirthmeticSingle} obtained analogues
results to Proposition \ref{prop:l>5} and Corollary \ref{cor:dist>=00003D6},
for the action of C+D on the projective plane.
\end{rem}

For $v\in L_{\mathcal{T}}$, let us say that two vertices $u,w$ are
in the same \emph{$v$-clan} if they are in $S_{6}(v)$, the 6-sphere
around $v$, and have a common grandfather in the 4-sphere of $v$
(in other words, $\dist(u,w)\leq4$). Each 6-sphere $S_{6}(v)$ has
size $4\cdot3^{5}=972$, and is divided to $108$ $v$-clans of size
$9$ each.
\begin{thm}
\label{thm:orbits} Let $\Lambda$ be the subgroup of $\Gamma$ generated
by $H$ and $\mathcal{D}$.
\begin{enumerate}
\item The $\Gamma$-orbit of any $v\in\Gamma v_{0}$ contains a unique member
of each $v$-clan.
\item The same holds for the $\Lambda$-orbit of $v\in\Gamma v_{0}$.
\item The $\Lambda$-orbit and $\Gamma$-orbit of $v_{0}$ are equal $\left(\Lambda v_{0}=\Gamma v_{0}\right)$.
\end{enumerate}
\end{thm}

\begin{proof}
We prove (1) and (2) simultaneously: note that $\ell\left(H\right)=6$,
so that $Hv_{0}\in S_{6}(v_{0})$. As $C$ fixes $v_{0}$, it acts
on $S_{6}(v_{0})$, and the stabilizer of $Hv_{0}$ is $Stab_{C}(Hv_{0})=C\cap HCH^{-1}$,
which has size $18^{2}$. Thus, the $C$-orbit of $Hv_{0}$ is of
size 
\[
\frac{|C|}{|C\cap HCH^{-1}|}=\frac{3!\cdot18^{3}}{18^{2}}=108,
\]
and as $\{H\}\cup C\subseteq\Lambda\subseteq\Gamma$, both $\Lambda$
and $\Gamma$ take $v_{0}$ to at least $108$ vertices in $S_{6}(v_{0})$.
On the other hand, by Corollary \ref{cor:dist>=00003D6}, $\Gamma$
(and thus also $\Lambda$) cannot take $v_{0}$ to two members of
the same $v_{0}$-clan, since these have distance $\leq4$ between
them. Since there are $108$ $v_{0}$-clans, $\Gamma$ and $\Lambda$
must obtain one vertex from each. By translation, the same holds for
a general $v$ in the $\Gamma$-orbit of $v_{0}$.

(3) For $v\in\Gamma v_{0}$, we want to show that $v\in\Lambda v_{0}$,
and we proceed by induction on $n=\dist(v,v_{0})$, where $n=0$ is
clear. Assume that $n>0$, which by Cor.\ \ref{cor:dist>=00003D6}
implies $n\geq6$. Let $w\in S_{6}(v)$ be the vertex on the path
from $v$ to $v_{0}$, so that $\dist(w,v_{0})=n-6$:\\
\hspace*{\fill}\includegraphics[scale=2]{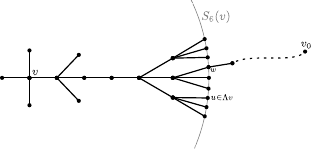}\hspace*{\fill}\\
Let $u$ be the unique member of the $v$-clan of $w$ which is in
the $\Lambda$-orbit of $v$, so that 
\begin{equation}
\dist(u,v_{0})\leq\dist(u,w)+\dist(w,v_{0})\leq4+n-6=n-2.\label{eq:reduct}
\end{equation}
Since $u\in\Lambda v\subseteq\Gamma v=\Gamma v_{0}$, we can use the
induction hypothesis to conclude that $u\in\Lambda v_{0}$, and thus
also $v\in\Lambda v_{0}$.
\end{proof}
We can now prove our first main theorem:
\begin{thm}
\label{thm:main}
\begin{enumerate}
\item The C+D gate set $\left\{ H\right\} \cup\mathcal{D}$ generates $\Gamma$.
\item There is an efficient algorithm to solve the word problem in $\Gamma$
w.r.t.\ $\left\{ H\right\} \cup\mathcal{D}$.
\end{enumerate}
\end{thm}

\begin{proof}
(1) This follows from Theorem \ref{thm:orbits}(3) and Lemma \ref{lem:C}(2)
by a general principle: If $G\curvearrowright X$, and $H\leq G$
is such that $Gx=Hx$ and $\mathrm{Stab}_{G}\left(x\right)\subseteq H$
for some $x\in X$, then $G=H$. Indeed, for $g\in G$ there must
exist $h\in H$ such that $gx=hx$, hence $h^{-1}g\in\mathrm{Stab}_{G}\left(x\right)$,
and therefore $g=hh^{-1}g\in H$. The case at hand is that of $G=\Gamma$,
$X=V_{\mathcal{T}}$, $H=\Lambda$ and $x=v_{0}$.

(2) Let $\gamma\in\Gamma$. By the proof of Theorem \ref{thm:orbits},
there exists $c_{1}\in C$ such that $\dist(\gamma c_{1}Hv_{0},v_{0})\leq\dist(\gamma v_{0},v_{0})-2$.
Furthermore, computing that $Stab_{\mathcal{D}}(Hv_{0})=\mathcal{D}\cap HCH^{-1}=\left\langle -\xi I,\diag(1,\xi^{3},\xi^{9})\right\rangle ,$we
obtain that the $\mathcal{D}$-orbit of $Hv_{0}$ is of size $|\mathcal{D}|/(18\cdot3)=108$.
Thus, it coincides with the $C$-orbit of $Hv_{0}$, so we can assume
that $c_{1}$ is in fact in $\mathcal{D}$. We can continue in this
manner to find $c_{2},\ldots,c_{r}\in\mathcal{D}$ such that 
\[
\dist(\gamma c_{1}Hc_{2}H\ldots c_{j}Hv_{0},v_{0})\leq\dist(\gamma c_{1}Hc_{2}H\ldots c_{j-1}Hv_{0},v_{0})-2\qquad\left(\forall1\leq j\leq r\right).
\]
In particular, $\dist(\gamma c_{1}Hc_{2}H\ldots c_{r}Hv_{0},v_{0})=0$,
so that $c_{r+1}:=\gamma c_{1}Hc_{2}H\ldots c_{r}H\in C$. We obtain
$\gamma=c_{r+1}H^{-1}c_{r}^{-1}H^{-1}c_{r-1}^{-1}\ldots H^{-1}c_{1}^{-1}$,
and $c_{r+1}$ can be expressed using C+D via (\ref{eq:C-M3-D}),
(\ref{eq:M3-fromCD}). To actually find the $c_{j}$ which shortens
the distance to $v_{0}$, one can choose $108$ representatives for
$\nicefrac{\mathcal{D}}{\mathcal{D}\cap HCH^{-1}}$, and try each
of them. However, there is also a way to do this in a single step
using the $p$-adic Iwasawa decomposition -- this is described in
\cite[§3.3]{Evra2018RamanujancomplexesGolden}.
\end{proof}

\section{Thin groups \protect\label{sec:Amalgamation-thin}}

In this section we move back to the setting of general S-arithmetic
subgroups of $PU(3)$, and study the question of Thinness. We present
criteria for Zariski denseness and for being of infinite index, and
show in §\ref{subsec:-Thinness-of-Clifford+T} that both apply to
the C+T gates, so that they generate a thin matrix group in $PU(3)$,
and in fact, even in $PU(3)^{3}$.

Recall that $\Sigma\subseteq PU(3)$ is called a universal gate set
if and only if the group $\Delta=\left\langle \Sigma\right\rangle $
is dense in $PU(3)$. When the entries of $\Sigma$ are in $\overline{\mathbb{Q}}$
, this is equivalent to $\Delta$ (embedded in $PU(3)$ in a fixed
manner) being Zariski dense in $PGL_{3}(\mathbb{C})$ (see \cite{bourgain2012spectral}).
A Zariski dense subgroup $\Delta$ of an S-arithmetic group is called
a \emph{thin matrix group }if its index is infinite (for more on thin
matrix groups see \cite{breuillard2014thin}). 

The next Proposition which relies on the work of Weigel \cite{Weigel1995certainclassFrattini},
gives a useful criterion for proving that a subgroup of a $\Pi$-arithmetic
group is Zariski dense in $G'=SU_{3}^{E,\Phi}$. As we have noted,
the latter is equivalent to topological density in $G'\left(F_{\varepsilon_{i}}\right)\cong SU\left(3\right)$,
for any $i=1,\ldots,d$, separately. What we shall prove is the stronger
property of denseness in the product group:
\begin{prop}
\label{prop:Weigel}Let $\ell\geq5$ be a rational prime, unramified
in $E$, with prime decomposition $\left(\ell\right)=\mathfrak{p}_{1}\cdot\ldots\cdot\mathfrak{p}_{f}$
in $\mathcal{O}$, such that $G'\left(F_{\mathfrak{p}_{i}}\right)$
is unramified for every $i$. If $\Pi$ is an $\mathcal{O}$-prime
coprime to $\ell$, and $\Delta\leq G'\big(\mathcal{O}\big[\frac{1}{\Pi}\big]\big)$
satisfies $\Delta\mod{\ell}=G'\left(\mathcal{O}/\ell\right)$, then
$\Delta$ is dense in $\prod_{i=1}^{d}G'\left(F_{\varepsilon_{i}}\right)\cong SU\left(3\right)^{d}$
(both in Zariski and archimedean topology).
\end{prop}

\begin{proof}
Let $\mathcal{G}=\mathrm{Res}_{\mathcal{O}/\mathbb{Z}}G'$ be Weil's
restriction of scalars of $G'$ from $\mathcal{O}$ to $\mathbb{Z}$,
namely, $\mathcal{G}$ is the group scheme over $\mathbb{Z}$ defined
by $\mathcal{G}\left(A\right)=G'\left(A\otimes_{\mathbb{Z}}\mathcal{O}\right)$
for any $\mathbb{Z}$-algebra $A$. In particular, $\mathcal{G}\left(\mathbb{R}\right)=G'\left(\mathbb{R}\otimes_{\mathbb{Z}}\mathcal{O}\right)=\prod_{i=1}^{d}G'\left(F_{\varepsilon_{i}}\right)$,
$\mathcal{G}\left(\mathbb{Z}_{\ell}\right)=G'\left(\mathbb{Z}_{\ell}\otimes_{\mathbb{Z}}\mathcal{O}\right)=\prod_{i=1}^{f}G'\left(\mathcal{O}_{\mathfrak{p}_{i}}\right)$,
$\mathcal{G}\left(\mathbb{Z}/\ell\right)=G'\left(\mathcal{O}/\ell\right)=\prod_{i=1}^{f}G'(\mathcal{O}/\mathfrak{p}_{i})$
and $\Delta\leq G'\big(\mathcal{O}\big[\frac{1}{\Pi}\big]\big)\leq G'\left(\mathcal{O}\big[\tfrac{1}{p}\big]\right)=\mathcal{G}\big(\mathbb{Z}\big[\frac{1}{p}\big]\big)$,
where $p$ is the rational prime below $\Pi$. By assumption, $\Delta$
satisfies $\Delta\mod{\ell}=\mathcal{G}\left(\mathbb{Z}/\ell\right)$,
and we shall prove that $\Delta$ is Zariski dense in $\mathcal{G}(\mathbb{Q})$.
Since $\ell$-adic topology is finer than the Zariski topology it
suffices to prove the density in the $\ell$-adic topology. Let $\widehat{\Delta}\leq\mathcal{G}\left(\mathbb{Z}_{\ell}\right)$
be the completion of $\Delta$ in the $\ell$-adic topology, so that
$\widehat{\Delta}\mod{\ell}=\mathcal{G}\left(\mathbb{Z}/\ell\right)$,
and we want to show that $\widehat{\Delta}=\mathcal{G}\left(\mathbb{Z}_{\ell}\right)$,
as the latter is Zariski dense in $\mathcal{G}(\mathbb{Q}_{\ell})$
(see \cite{Lubotzky1999GenerationSLn}). By the work of Wiegel on
Frattini extensions \cite{Weigel1995certainclassFrattini}, for $r_{i}\colon G'\left(\mathcal{O}_{\mathfrak{p}_{i}}\right)\to G'\left(\mathcal{O}/\mathfrak{p}_{i}\right)\colon g\mapsto g\negmedspace\mod{\mathfrak{p}}_{i}$
we know that $\ker r_{i}$ is contained in the Frattini subgroup of
$G'\left(\mathcal{O}_{\mathfrak{p}_{i}}\right)$ (for $\mathfrak{p_{i}}$
split see \cite[Cor.\ A]{Weigel1995certainclassFrattini}, and for
$\mathfrak{p}_{i}$ inert see \cite[Lem.\ 3.7]{Dettweiler2006threedimensionalGalois}).
Since the Frattini of the product of groups is the product of their
Frattinis it follows that $\ker r_{\ell}$, where $r_{\ell}\colon\mathcal{G}(\mathbb{Z}_{\ell})\to\mathcal{G}(\mathbb{Z}/\ell)$,
$r_{\ell}(g)=g\mod{\ell}$, is contained in the Frattini subgroup
of $\mathcal{G}(\mathbb{Z}_{\ell})$. By \cite[Cor.\ 1]{Gruenberg1967Projectiveprofinitegroups},
if $H\leq\mathcal{G}\left(\mathbb{Z}_{\ell}\right)$ is an open subgroup
such that $H\cdot\Phi=\mathcal{G}\left(\mathbb{Z}_{\ell}\right)$,
where $\Phi$ is the Frattini subgroup of $\mathcal{G}\left(\mathbb{Z}_{\ell}\right)$,
then $H=\mathcal{G}\left(\mathbb{Z}_{\ell}\right)$. Since $\widehat{\Delta}\leq\mathcal{G}\left(\mathbb{Z}_{\ell}\right)$
is an open subgroup, this shows that $\widehat{\Delta}\mod{\ell}=\mathcal{G}\left(\mathbb{Z}/\ell\right)$
implies $\widehat{\Delta}=\mathcal{G}\left(\mathbb{Z}_{\ell}\right)$,
as claimed. It follows that $\Delta$ is Zariski dense in $\mathcal{G}(\mathbb{Q})$,
and since $\mathcal{G}(\mathbb{Q})$ is archimedean-dense in $\mathcal{G}(\mathbb{R})$
by weak approximation \cite{Platonov1994Algebraicgroupsand}, it is
also Zariski-dense in it, so that $\Delta$ is Zariski dense in $\mathcal{G}(\mathbb{R})$.
Finally, let $\overline{\Delta}$ be the archimedean closure of $\Delta$
in $\mathcal{G}(\mathbb{R})$; as $\overline{\Delta}$ is a compact
Lie subgroup of $\mathcal{G}(\mathbb{R})$, it is algebraic by the
classical work of Tannaka \cite{Tannaka1939UberdenDualitatssatz},
so by the Zariski denseness we have $\overline{\Delta}=\mathcal{G}(\mathbb{R})$.
\end{proof}
The next Proposition gives us a simple criterion for certain subgroups
of groups acting on trees to be of infinite index. Let $G$ be a
group which acts on an infinite tree $\mathcal{T}$ without inverting
edges\footnote{This is true for any subgroup of $p$-adic $U_{3}$, but can be arranged
in general by passing to the barycentric subdivision of the tree.}, and let $Y=\mathcal{T}/G$ be the quotient graph. Choose a spanning
tree $T\subseteq Y$, and a section $j\colon Y\rightarrow\mathcal{T}$
of the quotient map $\mathcal{T}\rightarrow Y$, such that $j|_{T}$
is an isomorphism. This $j$ is not necessarily a graph morphism,
as $j(\{v,w\})$ can be different from $\{j(v),j(w)\}$ for vertices
$v,w$: every edge $e$ in $Y$ which is not in $T$ is of the form
$\{v,w\}$ with $j(e)=\{j(v),g_{e}j(w)\}$ for some $g_{e}\in G$.
The same holds for $e$ which is in $T$ as well, with $g_{e}=1$.
Denoting $G_{v}=\Stab_{G}\left(j(v)\right)$, it is not hard to see
that $G$ is generated by 
\begin{equation}
\left\{ g_{e}\,\middle|\,e\in E_{Y\backslash T}\right\} \cup\bigcup\nolimits_{v\in V_{Y}}G_{v},\label{eq:BS-gens}
\end{equation}
and Bass-Serre theory (see \cite{serre1980trees}) goes further to
give an explicit presentation of $G$. It defines a graph of groups
$\left(G,Y\right)$, where $G_{v}$ is as above, $G_{e}=\Stab_{G}(j(e))$
for every $e\in E_{Y}$, and whenever $j(e)=\{j(v),g_{e}j(w)\}$ there
are inclusion maps $G_{e}\subseteq G_{v}$ (naturally) and $G_{e}\hookrightarrow G_{w}$
via $x\mapsto g_{e}^{-1}xg_{e}$. By \cite[§5.4]{serre1980trees},
one then has $G=\pi_{1}\left(G,Y,T\right)$.
\begin{prop}
\label{prop:B-S}In the settings above, let $S_{v}\subseteq G_{v}$
be some subset for each $v\in V_{Y}$, and 
\[
K=\left\langle \left\{ g_{e}\,\middle|\,e\in E_{Y\backslash T}\right\} \cup\bigcup\nolimits_{v\in V_{Y}}S_{v}\right\rangle .
\]
If there exists $v_{0}\in V_{Y}$ such that $\left\langle S_{v_{0}}\cup\bigcup_{v_{0}\in e}G_{e}\right\rangle \lneq G_{v_{0}}\lneq G$,
then $\left[G:K\right]=\infty$.
\end{prop}

\begin{proof}
Let $\left(K,Y\right)$ be the graph of groups with $K_{*}=G_{*}$
for $*\in V_{Y}\cup E_{Y}$, except for $K_{v_{0}}=\left\langle S_{v_{0}}\cup\bigcup_{v_{0}\in e}G_{e}\right\rangle \lneq G_{v_{0}}$,
and with the inclusion maps restricted from those of $\left(G,Y\right)$.
We then have $K\leq\pi_{1}\left(K,Y,T\right)$, so it is enough to
prove that the latter is of infinite index in $G=\pi_{1}\left(G,Y,T\right)$.
Let $\overline{Y}$ be the graph obtained by adjoining to $Y$ a new
vertex $v_{\infty}$ and an edge $e_{\infty}=\left\{ v_{0},v_{\infty}\right\} $,
and let $\left(\overline{K},\overline{Y}\right)$ be the graph of
groups extending $\left(K,Y\right)$ by $\overline{K}_{v_{\infty}}=G_{v_{0}}$,
$\overline{K}_{e_{\infty}}=K_{v_{0}}$, with the natural inclusion
maps. Taking $\overline{T}=T\cup\left\{ e_{\infty}\right\} $, we
observe that 
\[
G_{v_{0}}*_{K_{v_{0}}}\pi_{1}\left(K,Y,T\right)=\pi_{1}\left(\overline{K},\overline{Y},\overline{T}\right)\cong\pi_{1}\left(G,Y,T\right)=G
\]
where the isomorphism is obtained by contracting the edge $e_{\infty}$.
It now follows from the other direction of Bass-Serre theory \cite[§5.3]{serre1980trees}
that $G$ acts on a biregular tree of degrees $\left[G_{v_{0}}:K_{v_{0}}\right]$
and $\left[\pi_{1}(K,Y,T):K_{v_{0}}\right]$, transitively on each
side of the tree, with respective vertex stabilizers $G_{v_{0}}$
and $\pi_{1}(K,Y,T)$. This tree is infinite since we have assumed
$\left[G_{v_{0}}:K_{v_{0}}\right]>1$, and $\left[\pi_{1}(K,Y,T):K_{v_{0}}\right]=1$
would give $G=G_{v_{0}}$. It thus follows that $\left[G:\pi_{1}(K,Y,T)\right]=\infty$
as claimed.
\end{proof}

\subsection{\protect\label{subsec:-Thinness-of-Clifford+T}Thinness of Clifford+T}

In this section we specialize again to the extended Clifford gates,
so that $F,E,\pi,\Pi,\varepsilon_{i},\varphi$ are as in Section \ref{sec:Synthesis-and-arithmeticity}.
Let $\overline{G}=G/Z$, where $Z$ is the center of $G$, which is
a projective unitary group scheme, i.e. $\overline{G}\left(F_{\varepsilon_{1}}\right)=PU(3)$.
Since $Z\left(F_{\Pi}\right)$ acts trivially on $\mathcal{T}$, we
get $\overline{G}\left(F_{\Pi}\right)\curvearrowright\mathcal{T}$.
In this section we denote $\Gamma=\overline{G}\left(\mathcal{O}_{F}\left[\frac{1}{\Pi}\right]\right)\curvearrowright\mathcal{T}$. 

From Proposition \ref{prop:level_dist} we obtain that $Hv_{0}$ is
of distance $\ell(H)=6$ from $v_{0}$. Labeling the vertices on the
path from $v_{0}$ to $Hv_{0}$ by $v_{0},\ldots,v_{6}=Hv_{0}$, we
denote $C_{j}=\mathrm{Stab}_{\Gamma}(v_{j})$ for $j=0,1,2,3$, and
\begin{equation}
C_{D}=\mathrm{Stab}_{\Gamma}\left(D\right)=C_{0}\cap C_{3},\quad\text{where}\quad D=\underset{v_{0}}{\bullet}-\underset{v_{1}}{\bullet}-\underset{v_{2}}{\bullet}-\underset{v_{3}}{\bullet}.\label{eq:D-domain}
\end{equation}
Note that since we moved to the projective group, in this section
$\Gamma$ denotes the group $\Gamma$ of Section \ref{sec:Synthesis-and-arithmeticity}
modulo its center $\left\langle -\xi\right\rangle $.
\begin{lem}
\label{lem:stabilizers-transitive} Denote $N(v_{i})=\left\{ u\in V_{\mathcal{T}}\,\middle|\,\mathrm{dist}\left(u,v_{i}\right)=1\right\} $
for $i=0,1,2,3$, and $N'(v_{i})=N(v_{i})\setminus\left\{ v_{i-1}\right\} $
for $i=1,2,3$. Then:
\begin{enumerate}
\item $C_{0}$ acts transitively on $N(v_{0})$.
\item $C_{0}\cap C_{i}$ acts transitively on $N'(v_{i})$ for $i=1,2,3$.
\item $C_{3}$ acts transitively on $N(v_{3})$. 
\end{enumerate}
\end{lem}

\begin{proof}
By the proof of Theorem \ref{thm:orbits}(1), combined with the fact
that each $v_{0}$-clan is determined by its common grandfather in
$S_{4}(v_{0})$, we get that $C_{0}$ acts transitively on $S_{4}(v_{0})$.
This implies that $C_{0}$ acts transitively also on $S_{i}(v_{0})$,
for any $i=1,2,3$. Since $N(v_{0})=S_{1}(v_{0})$ and $N'(v_{i})\subset S_{i+1}(v_{0})$
for $i=1,2,3$, we get that $C_{0}$ acts transitively on $N(v_{0})$
and on $N'(v_{i})$ for $i=1,2,3$, in particular proving the first
claim. Note that if $c\in C_{0}$ is such that $cv_{i+1}\in N'(v_{i})$,
then $cv_{i}=v_{i}$, hence $c\in C_{i}$, and therefore $C_{0}\cap C_{i}$
acts transitively on $N'(v_{i})$ for $i=1,2,3$, proving the second
claim. Next, we note that $H^{2}\in C_{0}$, so that $H$ interchanges
$v_{0}$ and $Hv_{0}$, and thus reverses the path between them. In
particular, $H$ fixes $v_{3}$, and takes $v_{2}$ to $v_{4}\in N'(v_{3})$.
Combined with the fact that $C_{0}\cap C_{3}$ acts transitively on
$N'(v_{3})$, we get that $C_{3}$ acts transitively on $N(v_{3})$.
\end{proof}
\begin{prop}
\label{prop:FD} The path $D$ forms a fundamental domain for $\Gamma\curvearrowright\mathcal{T}$.
\end{prop}

\begin{proof}
It follows from Lemma \ref{lem:stabilizers-transitive}(1) that all
the edges incident to $v_{0}$ are in the $\Gamma$-orbit of $\left\{ v_{0},v_{1}\right\} $,
and from (3) that those incident to $v_{3}$ are in the orbit of $\left\{ v_{2},v_{3}\right\} $.
From (2) we obtain that for $i=1,2$ the edges connecting $v_{i}$
to a vertex in $N'(v_{i})$ are in the orbit of $\left\{ v_{i},v_{i+1}\right\} $.
Thus, $\Gamma D=\mathcal{T}$, and had some $\gamma\in\Gamma$ taken
$v_{i}\in D$ to $v_{j}\in D$ for $i\neq j$, then we would have
\[
\dist\left(v_{0},\gamma v_{0}\right)\leq\dist\left(v_{0},v_{j}\right)+\dist\left(v_{j},\gamma v_{0}\right)=\dist\left(v_{0},v_{j}\right)+\dist\left(v_{i},v_{0}\right)\leq5
\]
contradicting Corollary \ref{cor:dist>=00003D6}. 
\end{proof}
Next we wish to describe the stabilizers of the vertices and edges
in the fundamental domain. 
\begin{prop}
\label{prop:FD-stabilizers}We have $C_{D}\leq C_{2}\leq C_{1}\leq C_{0}$,
\begin{align}
C_{j} & =\left\{ c\in C_{0}\,\middle|\,\ell\left(H^{-1}cH\right)\leq12-2j\right\} \qquad\left(j=0,1,2\right),\label{eq:Cj}\\
C_{D} & =\left\{ c\in C_{0}\,\middle|\,\ell\left(H^{-1}cH\right)\leq6\right\} =S_{3}\ltimes\left\langle \left(\begin{smallmatrix}1\\
 & \zeta_{3}\\
 &  & 1
\end{smallmatrix}\right),\left(\begin{smallmatrix}1\\
 & 1\\
 &  & \zeta_{3}
\end{smallmatrix}\right)\right\rangle \label{eq:CD}
\end{align}
$C_{3}=\langle\{H\}\cup C_{D}\rangle$ and $C_{D}=C_{2}\cap C_{3}$.
Their sizes and some minimal generating sets are given by:
\[
\xymatrix@R=2pt{\stackrel[1944]{v_{0}}{\bullet}\ar@{-}^{486}[rr] &  & \stackrel[486]{v_{1}}{\bullet}\ar@{-}^{162}[rr] &  & \stackrel[162]{v_{2}}{\bullet}\ar@{-}^{54}[rr] &  & \stackrel[216]{v_{3}}{\bullet}\ar@{.}[r] & \overset{v_{4}}{\bullet}\ar@{.}[r] & \overset{v_{5}}{\bullet}\ar@{.}[r] & \overset{v_{6}}{\bullet}}
\]
\end{prop}

\begin{proof}
We have $C_{2}\leq C_{1}\leq C_{0}$, since any $\gamma$ which violates
one of these inclusions must carry $v_{0}$ to a vertex of distance
smaller than $6$, contradicting Corollary~\ref{cor:dist>=00003D6}.
Since $C_{D}=C_{0}\cap C_{1}\cap C_{2}\cap C_{3}$ and $C_{2}\leq C_{1}\leq C_{0}$,
we also obtain $C_{D}=C_{2}\cap C_{3}$. Since $C_{0}$ is $C$ modulo
its center $\left\langle -\xi\right\rangle $, we obtain from Lemma
\ref{lem:C}(1) that $\left|C_{0}\right|=\frac{|C|}{18}=1944$, and
the sizes of $|C_{1}|,|C_{2}|,|C_{D}|,|C_{3}|$ can now be easily
computed by orbit-stabilizer considerations: Lemma \ref{lem:stabilizers-transitive}
and Proposition \ref{prop:FD} determine the orbits of $C_{i}$ acting
on its adjacency edges, and we already know that $C_{2}\leq C_{1}\leq C_{0}$.

To prove (\ref{eq:Cj}) we first observe that $\ell\left(H^{-1}cH\right)=\dist(Hv_{0},cHv_{0})=\dist\left(v_{6},cv_{6}\right)$
by Proposition \ref{prop:level_dist}, and in particular for every
$c\in C_{0}$ we have $\ell\left(H^{-1}cH\right)\leq12$. Now, if
$c\in C_{0}\backslash C_{1}$ then $\dist\left(v_{6},cv_{6}\right)=12$,
and similarly if $c\in C_{1}\backslash C_{2}$ then $\dist\left(v_{6},cv_{6}\right)=10$,
which gives (\ref{eq:Cj}). For the same reason, if $c\in C_{2}\backslash C_{3}$
then $\dist\left(v_{6},cv_{6}\right)=8$, which gives the first equality
in (\ref{eq:CD}), and the second one is by explicit computation.
Finally, we already know that $H$ fixes $v_{3}$, and computing that
$\left\langle \{H\}\cup C_{D}\right\rangle $ has size $216$ we obtain
$C_{3}=\left\langle H\cup C_{D}\right\rangle $.
\end{proof}
\begin{rem}
We give here, without proofs, minimal generating sets for the various
stabilizers (this will be not used in the paper): $C_{0}=\left\langle \text{\ensuremath{\left(\begin{smallmatrix}  &  1\\
  &   &  1\\
 1 
\end{smallmatrix}\right)},\ensuremath{\left(\begin{smallmatrix}1\\
  &   &  \xi\\
  &  -\xi
\end{smallmatrix}\right)}}\right\rangle $, $C_{1}=\left\langle \text{\ensuremath{\left(\begin{smallmatrix}1\\
  &   &  \xi\\
  &  1 
\end{smallmatrix}\right)},\ensuremath{\left(\begin{smallmatrix}  &   &  \zeta_{3}\\
 1\\
  &  1 
\end{smallmatrix}\right)}}\right\rangle $, $C_{2}=\left\langle \text{\ensuremath{\left(\begin{smallmatrix} 1\\
  &   &  \zeta_{3}\\
  &  \overline{\zeta_{3}} 
\end{smallmatrix}\right)},\ensuremath{\left(\begin{smallmatrix}  &  1 \\
  &   &  \xi\\
 \overline{\xi} 
\end{smallmatrix}\right)},\ensuremath{\left(\begin{smallmatrix}  &  1 \\
  &   &  1\\
 \zeta_{3} 
\end{smallmatrix}\right)}}\right\rangle $, $C_{D}=\left\langle H,\text{\ensuremath{\left(\begin{smallmatrix}1\\
  &  1\\
  &   &  \zeta_{3} 
\end{smallmatrix}\right)}}\right\rangle $.
\end{rem}

A nice corollary of our analysis is a structure theorem for $\Gamma$:
\begin{cor}
\label{cor:Bass-Serre} $\Gamma=C_{0}*_{C_{D}}C_{3}$.
\end{cor}

\begin{proof}
From Bass-Serre theory \cite{serre1980trees} combined with the fact
that $C_{D}\leq C_{2}\leq C_{1}\leq C_{0}$, we obtain
\[
\Gamma=C_{0}*_{C_{0}\cap C_{1}}C_{1}*_{C_{1}\cap C_{2}}C_{2}*_{C_{2}\cap C_{3}}C_{3}=C_{0}*_{C_{1}}C_{1}*_{C_{2}}C_{2}*_{C_{D}}C_{3}=C_{0}*_{C_{D}}C_{3}.\qedhere
\]
\end{proof}
 We can now prove our second main result:
\begin{thm}
\label{thm:C+T-thin} The group $\Delta=\langle H,S,T\rangle$ generated
by the C+T gates is a thin matrix group in $\Gamma\leq PU(3)^{3}$.
\end{thm}

\begin{proof}
First we prove that $\Delta$ is of infinite index in $\Gamma$. For
the action of $\Gamma$ on $\mathcal{T}$, the fundamental domain
$\underset{v_{0}}{\bullet}-\underset{v_{1}}{\bullet}-\underset{v_{2}}{\bullet}-\underset{v_{3}}{\bullet}$
is isomorphic to the quotient $Y=\Gamma\backslash\mathcal{T}$, and
equals its own spanning tree. We take $S_{v_{0}}=\left\{ S,T\right\} $,
$S_{v_{1}}=S_{v_{2}}=\varnothing$ and $S_{v_{3}}=\left\{ H\right\} $,
and observe that $\ell\left(H^{-1}SH\right)=6$ and $\ell\left(H^{-1}TH\right)=8$
imply $S,T\in C_{1}$ by (\ref{eq:Cj}). Therefore, we have 
\[
\left\langle S_{v_{0}}\cup G_{\{v_{0},v_{1}\}}\right\rangle =\left\langle \{S,T\}\cup C_{1}\right\rangle =C_{1}\lneq C_{0}=G_{v_{0}},
\]
so that Proposition \ref{prop:B-S} implies $\left[\Gamma:\Delta\right]=\infty$.

Next, we prove that $\Delta$ is Zariski dense in $\Gamma\leq PU(3)^{3}$.
Recalling the notations of §\ref{sec:Synthesis-and-arithmeticity},
we take $\ell=19$, $\psi=1-\xi-\xi^{2}$ and $\Psi=\psi\overline{\psi}=5-\sigma^{2}$,
which satisfy $N_{E/\mathbb{Q}}(\psi)=N_{F/\mathbb{Q}}(\Psi)=\ell$.
Denoting $\mathfrak{p}_{i}=\left(\varphi^{i}(\Psi)\right)$, we have
$\left(\ell\right)=\mathfrak{p}_{0}\mathfrak{p}_{1}\mathfrak{p}_{2}$
in $\mathcal{\mathcal{O}}$, and these are all distinct, so that we
obtain $G'(\mathcal{O}/\ell)\cong\prod_{i=0}^{2}SU_{3}(\mathcal{O}/\mathfrak{p}_{i})$.
Furthermore, each $\mathfrak{p}_{i}$ splits into different factors
($\varphi^{i}(\psi)$ and $\varphi^{i}(\overline{\psi})$) in $E$.
This leads to $SU_{3}(\mathcal{O}/\mathfrak{p}_{i})\cong SL_{3}(\mathcal{O}_{E}/\varphi^{i}(\psi))\cong SL_{3}(\mathbb{F}_{19})$,
and in total $G'(\mathcal{O}/\ell)\cong SL_{3}(\mathbb{F}_{19})^{3}$,
which makes computations in $G'(\mathcal{O}/\ell)$ relatively feasible.
If $\vartheta_{i}\colon\mathcal{O}_{E}/\varphi^{i}(\psi)\overset{{\scriptscriptstyle \cong}}{\longrightarrow}\mathbb{F}_{19}$,
the isomorphism $\Theta\colon G'(\mathcal{O}/\ell)\cong SL_{3}(\mathbb{F}_{19})^{3}$
is given by $\Theta\left(A\right)=\left(\vartheta_{0}(A),\vartheta_{1}(A),\vartheta_{2}(A)\right)$.
For $i=0$, for example, since $\psi=1-\xi-\xi^{2}$, $m_{\xi}^{\mathbb{Q}}(x)=1+x^{3}+x^{6}$,
and $\gcd(1-x-x^{2},1+x^{3}+x^{6})=x+15$ in $\mathbb{F}_{19}[x]$,
the isomorphism $\vartheta_{0}$ is given explicitly by $\vartheta_{0}\colon\xi\mapsto4$,
yielding 
\[
\vartheta_{0}\left(H\!\negthickspace\negthickspace\mod{\mathfrak{p}}_{0}\right)=\left(\begin{smallmatrix}14 & 14 & 14\\
14 & 3 & 2\\
14 & 2 & 3
\end{smallmatrix}\right),\quad\vartheta_{0}\left(S/\xi\!\negthickspace\negthickspace\mod{\mathfrak{p}}_{0}\right)=\left(\begin{smallmatrix}5\\
 & 16\\
 &  & 5
\end{smallmatrix}\right),\quad\vartheta_{0}\left(T\!\negthickspace\negthickspace\mod{\mathfrak{p}}_{0}\right)=\left(\begin{smallmatrix}4\\
 & 1\\
 &  & 5
\end{smallmatrix}\right).
\]
Similar computations give $\vartheta_{1}(\xi)=16$ and $\vartheta_{2}(\xi)=9$,
and one can then verify using GAP \cite{GAP4} that $\left\{ \Theta(H\negthickspace\!\mod{\ell}),\,\Theta(S/\xi\negthickspace\!\mod{\ell}),\ \Theta(T\negthickspace\!\mod{\ell})\right\} $
indeed generates $SL_{3}(\mathbb{F}_{19})^{3}\cong G'(\mathcal{O}/\ell)$.
By Proposition \ref{prop:Weigel} we conclude that $\left\langle H,S/\xi,T\right\rangle $
is Zariski dense in $SU_{3}^{E,\Phi}$, and since $PSU(3)=PU(3)$,
this implies that $\left\langle H,S,T\right\rangle $ is Zariski dense
in $PU(3)^{3}$.
\end{proof}

\section{Covering rate \protect\label{sec:Covering-rate}}

In this section we study the covering rate of families of finite sets
of points in the projective unitary group $PU(3)=U(3)/U(1)$, or product
of several copies of this group. Note that we do not expect to be
dense in $U(3)$ itself: for example $\Gamma$ which is generated
by the C+D gates is not dense in $U(3)$ since $\det\left(\Gamma\right)=\left\langle -\xi\right\rangle $
whereas $\det(U(3))=U(1)$. However, unitary gates are invariant to
scaling, so from the point of view of quantum gates we can move to
study $PU(3)$, in which $\Gamma$ is dense (by strong approximation
for S-arithmetic groups).
\begin{defn}[{\cite[Def. 2.8]{Evra2018RamanujancomplexesGolden,sarnak2015letter}}]
\label{def:aoac}A sequence of finite sets $\left\{ X_{r}\right\} _{r}$
in a compact Lie group \textbf{$\boldsymbol{L}$}, $|X_{r}|\rightarrow\infty$,
is said to be an almost-optimal almost-cover (a.o.a.c.)\ of $L$,
if there exists a polynomial $p\left(x\right)$ such that 
\[
\mu\left(\boldsymbol{L}\setminus B\left(X_{r},\varepsilon_{r}\right)\right)\rightarrow0,\qquad\mbox{when}\qquad\varepsilon_{r}=\frac{p\left(\log|X_{r}|\right)}{|X_{r}|};
\]
here $\mu$ is the probability Haar measure on $\boldsymbol{L}$,
$B\left(X,\varepsilon\right)=\bigcup_{x\in X}B\left(x,\varepsilon\right)$
and $B\left(x,\varepsilon\right)\subset\boldsymbol{L}$ is the ball
of volume $\varepsilon>0$ around $x\in\boldsymbol{L}$, w.r.t.\ the
bi-invariant metric of $\boldsymbol{L}$ (which is $d\left(g,h\right)=\sqrt{1-\left|\mathrm{tr}\left(g^{*}h\right)\right|/3}$
in the case of $\boldsymbol{L}=PU(3)$).
\end{defn}

We assume again the general settings of Section \ref{sec:Amalgamation-thin},
and consider 
\[
\boldsymbol{L}=\overline{G}\left(F\otimes_{\mathbb{Z}}\mathbb{R}\right)=\overline{G}(F_{\varepsilon_{1}}\times\ldots\times F_{\varepsilon_{d}})=\overline{G}(F_{\varepsilon_{1}})\times\ldots\times\overline{G}(F_{\varepsilon_{d}})\cong PU(3)^{d},
\]
where the last isomorphism is due to the fact that $\Phi$ is totally-definite.
The group $\Gamma=\overline{G}\big(\mathcal{O}_{F}\big[\frac{1}{\Pi}\big]\big)$
is embedded in $\boldsymbol{L}$ via $\gamma\mapsto\left(\varepsilon_{1}(\gamma),\ldots,\varepsilon_{d}(\gamma)\right)$,
and our goal is to study the covering rate of $\L$ by various subsets
of $\Gamma$. In Definition \ref{def:Clozel-Hecke}, we present a
sequence of finite sets $\left\{ \Omega_{r}\right\} _{r}\subset\Gamma$,
$|\Omega_{r}|\rightarrow\infty$, which we call Clozel-Hecke points,
and in Corollary \ref{cor:Clozel-a.o.a.c.} we show that they are
an a.o.a.c.\ of $\boldsymbol{L}=PU(3)^{d}$. We then relate the sets
of Clozel-Hecke points to two other sequences of finite sets.

The group $\Gamma$ acts on the Bruhat-Tits tree $\mathcal{T}$ associated
with the $\Pi$-adic completion $\overline{G}_{\Pi}:=\overline{G}\left(F_{\Pi}\right)$.
We assume from now on that $\Pi$ is ramified in $E$, in which case
the tree $\mathcal{T}$ is $\left(p+1\right)$-regular where $p=N_{F/\mathbb{Q}}(\Pi)$
(not necessarily a prime).\footnote{When $\Pi$ is inert, $\mathcal{T}$ is a $(p^{3}+1,p+1)$-biregular
tree -- see \cite{Evra2018RamanujancomplexesGolden,Evra2022Ramanujanbigraphs}
for more details.} We begin the analysis by observing the space of $\Gamma$-equivariant
families of functions on $L^{2}\left(\L\right)$ indexed by $L_{\mathcal{T}}$
(the left vertices of $\mathcal{T}$), namely: 
\[
V=\left\{ \left(f_{v}\right)_{v\in L_{\mathcal{T}}}\in L^{2}\left(\L\right)^{L_{\mathcal{T}}}\,\middle|\,f_{\gamma v}\left(x\right)=f_{v}\left(\gamma^{-1}x\right),\;\forall\gamma\in\Gamma,\;v\in L_{\mathcal{T}},\;x\in\L\right\} .
\]
For any even $r\in\mathbb{N}$, we define the normalized $r$-th Hecke
operator 
\[
T_{r}\colon V\rightarrow V,\qquad(T_{r}f)_{v}=\frac{1}{d_{r}}\sum\nolimits_{w\in S_{r}\left(v\right)}f_{w},\quad\forall v\in L_{\mathcal{T}},
\]
where $S_{r}\left(v\right)$ denote the $r$-sphere in $\mathcal{T}$
around $v\in L_{\mathcal{T}}$, and $d_{r}:=|S_{r}\left(v\right)|=(p+1)\cdot p^{r-1}$.
By the Borel--Harish-Chandra theory $\Gamma$ is cocompact in $\overline{G}_{\Pi}\times\L$,
and as $\L$ is compact, $\Gamma$ is also cocompact in $\overline{G}_{\Pi}$,
hence the quotient $\Gamma\backslash\mathcal{T}$ is a finite graph.
Let $v_{0},\ldots,v_{h-1}$ be the vertices belonging to $L_{\mathcal{T}}$
in a (connected) fundamental domain for the action of $\Gamma$ on
$\mathcal{T}$ (this $h$ is called the \emph{class number }of $\overline{G}$),
denote $\Gamma_{i}=\mathrm{Stab}_{\Gamma}\left(v_{i}\right)$ for
$i=0,\ldots,h-1$ and $\mathfrak{m}=\sum_{i=0}^{h-1}\frac{1}{\left|\Gamma_{i}\right|}$.
We consider $V$ with the inner product 
\begin{equation}
\left\langle f,g\right\rangle =\frac{1}{\mathfrak{m}}\sum\limits_{i=0}^{h-1}\frac{1}{\left|\Gamma_{i}\right|}\int_{\L}f_{v_{i}}(x)\overline{g_{v_{i}}(x)}\,d\mu(x),\label{eq:Vinn}
\end{equation}
and denote $V_{0}=\one^{\bot}=\{\left(f_{v}\right)\in V\,|\,\sum_{i=0}^{h-1}\tfrac{1}{|\Gamma_{i}|}\int_{\L}f_{v_{i}}=0\}$,
which is a $T_{r}$-invariant space. The following Theorem relies
heavily on the theory of automorphic representations of $U_{3}$ and
the Ramanujan Conjecture, which was studied in depth in \cite[§4,5]{Evra2018RamanujancomplexesGolden}
and \cite[§7]{Evra2022Ramanujanbigraphs}. We give a concise proof,
as the relevant details comprise a large part of these papers; The
two novel features here is that we consider a ramified place $\Pi$,
and do not assume that the class number is $1$, as is done in \cite{Evra2018RamanujancomplexesGolden,Evra2022Ramanujanbigraphs}.
\begin{thm}
\label{thm:Hecke-Ramanujan} For any even $r\in\mathbb{N}$, $\left\Vert T_{r}|_{V_{0}}\right\Vert \leq\frac{r+1}{\sqrt{d_{r}}}$.
\end{thm}

\begin{proof}
Using $L_{\mathcal{T}}\cong\nicefrac{\overline{G}_{\Pi}}{K_{\Pi}}$,
$V$ can be identified with the space of $K_{\Pi}$-fixed vectors
in the right regular $\overline{G}_{\Pi}$-representation $\widetilde{V}:=L^{2}\left(\Gamma\backslash(\L\times\overline{G}_{\Pi})\right)$.
In fact, this is where the inner product (\ref{eq:Vinn}) comes from.
As $\Gamma$ is cocompact in $\L\times\overline{G}_{\Pi}$, we can
decompose $\widetilde{V}$ to its irreducible $\overline{G}_{\Pi}$-representations,
$\widetilde{V}=\widehat{\bigoplus}_{i}\rho_{i}$. Then, $V\cong\widetilde{V}^{K_{\Pi}}=\widehat{\bigoplus}_{i}\rho_{i}^{K_{\Pi}}$
as a module over the Hecke algebra $\mathcal{H}_{\overline{G}_{\Pi}}$,
and the latter contains $T_{r}$, so that $\mathrm{Spec}\left(T_{r}\right)=\overline{\bigcup_{i}\mathrm{Spec}(T_{r}|_{\rho_{i}^{K_{\Pi}}})}$.
For each $\rho_{i}$ and $\lambda\in\Spec(T_{r}|_{\rho_{i}^{K_{\Pi}}})$
we observe that:

\emph{(i)} if $\rho_{i}$ is \emph{one-dimensional} (and $\rho_{i}^{K_{\Pi}}\neq0$),
then it is trivial since $\overline{G}'_{\Pi}K_{\Pi}=PSU_{3}(F_{\Pi})U_{3}(\mathcal{O}_{F_{\Pi}})=\overline{G}_{\Pi}$,
and then $\lambda=1$. Furthermore, each $f\in\rho_{i}^{K_{\Pi}}$
is fixed under both $\overline{G}_{\Pi}$ and $\Gamma$, and $\Gamma$
is dense in $\L$ by strong approximation (which applies as $PU(3)=PSU(3)$,
and $\overline{G}_{\Pi}$ is non-compact). Thus such $f$ is constant,
so that $V_{0}$ consists entirely of infinite-dimensional representations.

\emph{(ii)} if $\rho_{i}$ is \emph{tempered} then it is weakly contained
in $L^{2}(\overline{G}_{\Pi})$ \cite{Haagerup1988}, which implies
that $\lambda$ is in the $L^{2}$-spectrum of $T_{r}$ acting on
the Bruhat-Tits tree $\mathcal{T}$. The action of $T_{r}$ on this
tree is by averaging over a sphere of radius $r$, and in general,
for a $(p+1)$-regular tree $\mathcal{T}_{p+1}$ and $r\ge1$, the
spectral radius of this operator equals\footnote{This can be shown in many ways, e.g.\ using Chebyshev polynomials
as in \cite{Nestoridi2021Boundedcutoffwindow}, Harish-Chandra $\Xi$
function as in \cite{Evra2018RamanujancomplexesGolden}, or spectral
analysis of the non-backtracking operator on edges as in \cite{Evra2022Ramanujanbigraphs}.}
\[
\left\Vert T_{r}\big|_{\mathcal{T}}\right\Vert =\frac{p^{\frac{r-2}{2}}\left(pr+p-r+1\right)}{(p+1)p^{r-1}}\le\frac{r+1}{\sqrt{d_{r}}}.
\]
Let $\overline{G}\left(\mathbb{A}_{F}\right)=\prod'_{v}\overline{G}\left(F_{v}\right)$
be the $F$-adelic group and observe the compact open subgroup $K=K_{\Pi}K^{\Pi}$,
where 
\[
K^{\Pi}=\prod\nolimits_{\Pi\ne v\nmid\infty}\overline{G}\left(\mathcal{O}_{F_{v}}\right).
\]
By the strong approximation property for $SU_{3}$, we obtain that
$\overline{G}(F)\L\overline{G}_{\Pi}K$ is a finite index normal subgroup
of $\overline{G}(\mathbb{A}_{F})$ \cite[Prop.\ 5.30]{Evra2022Ramanujanbigraphs}.
We also have $\Gamma=\overline{G}\left(F\right)\cap\L\overline{G}_{\Pi}K^{\Pi}$,
and together we get an embedding of $\overline{G}_{\Pi}$-sets: 
\[
\Gamma\backslash(\L\times\overline{G}_{\Pi})\hookrightarrow\overline{G}\left(F\right)\backslash\overline{G}\left(\mathbb{A}_{F}\right)/K^{\Pi}.
\]
This induces an embedding of $\mathcal{H}_{\overline{G}_{\Pi}}$-modules:
\[
V\cong L^{2}\left(\Gamma\backslash(\L\times\overline{G}_{\Pi})\right)^{K_{\Pi}}\hookrightarrow L^{2}\left(\overline{G}\left(F\right)\backslash\overline{G}\left(\mathbb{A}_{F}\right)\right)^{K},
\]
and every $\rho_{i}$ is a local factor at $\Pi$ of a $K$-spherical
$\overline{G}\left(\mathbb{A}_{F}\right)$-subrepresentation $\sigma$
of $L^{2}\left(\overline{G}\left(F\right)\backslash\overline{G}\left(\mathbb{A}_{F}\right)\right)$.
Since $K_{\Pi}\leq K$, such $\sigma$ is in particular Iwahori-spherical
at the prime $\Pi$ which is ramified in $E$, and it follows from
\cite[Thm.\ 7.3(1)]{Evra2022Ramanujanbigraphs}\footnote{To be precise, the theorem there is stated for $\mathbb{Q}$, but
the proof applies to any totally real number field.} that $\sigma$ is either $1$-dimensional, or has tempered local
factors. In particular, $\rho_{i}=\sigma_{\Pi}$ is either $1$-dimensional
or tempered, which implies that the spectrum of $T_{r}\big|_{V_{0}}$
is bounded by $\frac{r+1}{\sqrt{d_{r}}}$. Finally, Note that for
any $g\in\overline{G}_{\Pi}$, $gv_{0},\ldots,gv_{h-1}$ is also a
fundamental domain for the action of $\Gamma$ on $\mathcal{T}$ and
$\mathrm{Stab}_{\Gamma}\left(gv\right)=g\mathrm{Stab}_{\Gamma}\left(v\right)g^{-1}$
for any $v\in L_{\mathcal{T}}$, and a simple computation shows that
$\langle g.f_{1},g.f_{2}\rangle=\langle f_{1},f_{2}\rangle$ for any
$f_{1},f_{2}\in V$. This implies that the Hecke operator $T_{r}$
is self-adjoint since the distance function in $\mathcal{T}$ is symmetric,
and therefore its operator norm equals its spectral radius. 
\end{proof}
For any $w\in L_{\mathcal{T}}$, let $i\left(w\right)\in\left\{ 0,\ldots,h-1\right\} $
be such that $w\in\Gamma v_{i\left(w\right)}$ and let $\Gamma\left(w\right)=\left\{ \gamma\in\Gamma\,\middle|\,w=\gamma v_{i\left(w\right)}\right\} $,
which is a left coset of $\Gamma_{i(w)}$. To relate the spectral
theory of the Hecke operator $T_{r}$ to the covering rate of $\Gamma$,
we use the strategy from \cite{Clozel2002Automorphicformsand}, but
need to work a bit harder to accommodate the stabilizers $\Gamma_{i}$,
which Clozel assumes to be trivial.
\begin{defn}
\label{def:Clozel-Hecke} For $r\in\mathbb{N}$, define the set of
$r$-th Clozel-Hecke points in $\Gamma$ to be 
\[
\Omega_{r}=\bigcup_{i=0}^{h-1}\Omega_{r}\left(v_{i}\right),\qquad\Omega_{r}\left(v\right)=\bigcup_{w\in S_{r}\left(v\right)}\Gamma\left(w\right),
\]
and the normalized $r$-th Clozel-Hecke operator to be 
\[
\overline{T}_{r}\colon L^{2}\left(\L\right)\rightarrow L^{2}\left(\L\right),\qquad\overline{T}_{r}f\left(x\right)=\sum_{i=0}^{h-1}\sum_{w\in S_{r}(v_{i})}\sum_{\gamma\in\Gamma(w)}\frac{f\left(\gamma^{-1}x\right)}{\mathfrak{m}d_{r}|\Gamma_{i}||\Gamma(w)|}.
\]
We observe that $d_{r}\leq\left|\Omega_{r}\right|\leq\mathfrak{M}hd_{r}$,
where $\mathfrak{M}=\max_{i=0}^{h-1}|\Gamma_{i}|$, which shows in
particular that the growth rate of the Clozel-Hecke points is $|\Omega_{r}|=\Theta(p^{r})$.
The operator $\overline{T}_{r}$ is a weighted averaging operator,
in the sense that $\sum_{i}\sum_{w}\sum_{\gamma}\frac{1}{\mathfrak{m}d_{r}|\Gamma_{i}||\Gamma(w)|}=1$.
In particular $\overline{T}_{r}\one=\one$, and $L_{0}^{2}(\L):=\one^{\bot}=\{f\,|\,\smallint_{\L}\negmedspace f=0\}$
is $\overline{T}_{r}$-invariant. We denote 
\[
W_{\overline{T}_{r}}=\left\Vert \overline{T}_{r}\big|_{L_{0}^{2}(\L)}\right\Vert .
\]
\end{defn}

\begin{thm}
\label{thm:Clozel-Hecke-Ramanujan}For any even $r\in\mathbb{N}$
we have $W_{\overline{T}_{r}}\leq\frac{(r+1)\sqrt{h\mathfrak{M}}}{\sqrt{|\Omega_{r}|}}$.
\end{thm}

\begin{proof}
Let $\smash{L^{2}(\L)\stackrel[S]{J}{\longrightleftarrows}V}$ be
the following diagonal projection and averaging operators:
\[
\left(Jf\right)_{w}\left(x\right)=\tfrac{1}{|\Gamma(w)|}\sum\nolimits_{\gamma\in\Gamma(w)}f\left(\gamma^{-1}x\right),\qquad\left(Sf\right)\left(x\right)=\tfrac{1}{\mathfrak{m}}\sum\nolimits_{i=0}^{h-1}\tfrac{1}{|\Gamma_{i}|}f_{v_{i}}(x).
\]
The Clozel-Hecke operator is related to the Hecke operator by
\[
\overline{T}_{r}=S\circ T_{r}\circ J\colon L^{2}\left(\L\right)\rightarrow V\rightarrow V\rightarrow L^{2}\left(\L\right).
\]
Furthermore, $J$ and $S$ restrict to $\smash{L_{0}^{2}(\L)\longrightleftarrows V_{0}}$,
and in addition $\left\Vert J\right\Vert =\left\Vert S\right\Vert =1$
by Cauchy--Schwarz, (\ref{eq:Vinn}), and the fact that $\|J\one\|=\|S\one\|=1$.
In total, we obtain
\[
W_{\overline{T}_{r}}=\left\Vert \overline{T}_{r}|_{L_{0}^{2}(\L)}\right\Vert \leq\left\Vert S\right\Vert \left\Vert T_{r}|_{V_{0}}\right\Vert \left\Vert J\right\Vert \leq\left\Vert T_{r}|_{V_{0}}\right\Vert \leq\frac{r+1}{\sqrt{d_{r}}}\leq\frac{(r+1)\sqrt{h\mathfrak{M}}}{\sqrt{|\Omega_{r}|}}
\]
by Theorem \ref{thm:Hecke-Ramanujan} and $|\Omega_{r}|\le\mathfrak{M}hd_{r}$.
\end{proof}
 From Theorem \ref{thm:Clozel-Hecke-Ramanujan} and the work of \cite{Parzanchevski2018SuperGoldenGates}
we obtain the following almost-optimal almost-covering property:
\begin{cor}
\label{cor:Clozel-a.o.a.c.}The sequence of Clozel-Hecke points $\left\{ \Omega_{r}\right\} _{r}$
forms an a.o.a.c.\ of $\L$, for any sequence $\varepsilon_{r}=\omega\left(\frac{\log^{2}|\Omega_{r}|}{|\Omega_{r}|}\right)$
(namely, whenever $\frac{\varepsilon_{r}|\Omega_{r}|}{\log^{2}|\Omega_{r}|}\rightarrow\infty$).
\end{cor}

\begin{proof}
Note that $\overline{T}_{r}$ is an averaging convolution operator
supported on $\Omega_{r}$ in the sense of \cite[(3.2)]{Parzanchevski2018SuperGoldenGates},
and $W_{\overline{T}_{r}}^{2}\leq\frac{(r+1)^{2}h\mathfrak{M}}{|\Omega_{r}|}=O\Big(\frac{\log^{2}|\Omega_{r}|}{|\Omega_{r}|}\Big)$.
From $\varepsilon_{r}=\omega\left(\frac{\log^{2}|\Omega_{r}|}{|\Omega_{r}|}\right)$
we obtain that $W_{\overline{T}_{r}}^{2}/\varepsilon_{r}=o(1)$, which
implies $\mu\left(\boldsymbol{L}\setminus B\left(\Omega_{r},\varepsilon_{r}\right)\right)\rightarrow0$
by \cite[Prop. 3.1]{Parzanchevski2018SuperGoldenGates}.
\end{proof}
We also get the following covering results in the form of the Solovay-Kitaev
Theorem, with an optimal exponent $c=1$ and an explicit leading coefficient
arbitrarily close to $2$. In the language of \cite{sarnak2015letter},
this shows that the covering exponent of the Clozel-Hecke sequence
is at most $2$.
\begin{prop}
\label{prop:CH-covering}For any small enough $\varepsilon$, every
$g\in\L$ has an $\varepsilon$-approximation in $\Omega_{r}$ for
\[
r=2\log_{p}\tfrac{1}{\varepsilon}+3\log_{p}\log\tfrac{1}{\varepsilon}.
\]
\end{prop}

\begin{proof}
Since $\L$ is a Riemannian $8d$-manifold, we have $\lim_{\varepsilon\rightarrow0}\frac{\rad B(x,3^{8d}\varepsilon)}{\rad B(x,\varepsilon)}=3$,
so there exists $\delta>0$ such $\frac{\rad B(x,3^{8d}\varepsilon)}{\rad B(x,\varepsilon)}>2$
for $\varepsilon<\delta$. Taking $C=\max\left\{ 3^{8d},\rad(\L)/\delta\right\} $,
we obtain that $B\left(x,\varepsilon\right)$ contains the ball whose
radius is twice that of $B\left(x,C^{-1}\varepsilon\right)$. For
$r=2\log_{p}\tfrac{1}{\varepsilon}+3\log_{p}\log\frac{1}{\varepsilon}$
we get by Theorem \ref{thm:Clozel-Hecke-Ramanujan}
\[
\frac{W_{\overline{T}_{r}}}{\varepsilon}\leq\frac{(r+1)\sqrt{h\mathfrak{M}}}{\varepsilon\sqrt{|\Omega_{r}|}}<\frac{(r+1)\sqrt{h\mathfrak{M}}}{\varepsilon p^{r/2}}=\frac{\left(2\log_{p}\tfrac{1}{\varepsilon}+3\log_{p}\log\frac{1}{\varepsilon}+1\right)\sqrt{h\mathfrak{M}}}{\left(\log\tfrac{1}{\varepsilon}\right)^{3/2}}\overset{{\scriptscriptstyle \varepsilon\rightarrow0}}{\longrightarrow}0.
\]
In particular, for $\varepsilon$ small enough $W_{\overline{T}_{r}}$
is bounded by $C^{-1}\varepsilon$. By \cite[Cor. 3.2]{Parzanchevski2018SuperGoldenGates},
this implies that the Ball whose radius is twice that of $B(x,C^{-1}\varepsilon)$
around $\Omega_{r}$ covers $\L$, and by the choice of $C$ this
implies $\L=B\left(\Omega_{r},\varepsilon\right)$. 
\end{proof}
Next, we wish to study the covering rates of other sets in $\Gamma$,
for example, using word/circuit length in a chosen set of generators
as a measure of complexity. Let $\ell_{CH}\colon\Gamma\rightarrow\mathbb{N}$
be the \emph{Clozel-Hecke (CH) length}, defined by $\ell_{CH}\left(\gamma\right)=\min\left\{ r\,\middle|\,\gamma\in\Omega_{r}\right\} $.
Let us say that a function $\ell_{\times}\colon\Gamma\rightarrow\mathbb{N}$
is $(c,C,b)$-quasi-isometric (q.i.) to $\ell_{CH}$, where $C\geq c>0$
and $b>0$, if 
\[
c\cdot\ell_{CH}\left(\gamma\right)-b\leq\ell_{\times}\left(\gamma\right)\leq C\cdot\ell_{CH}\left(\gamma\right)+b,\qquad\forall\gamma\in\Gamma.
\]
Denote the balls of radius $r$ around $1$ in $\Gamma$ w.r.t.\ $\ell_{\times}$,
by 
\[
B_{r}^{\ell_{\times}}=\left\{ \gamma\in\Gamma\,\middle|\,\ell_{\times}\left(\gamma\right)\leq r\right\} .
\]
The following Proposition implies an almost-covering property for
sequences of balls w.r.t.\ length functions which are quasi-isometric
to the CH length function.
\begin{prop}
\label{prop:a.o.a.c-a.c.} If $\ell_{\times}$ is $(c,C,b)$-q.i.
to $\ell_{CH}$, then:
\begin{enumerate}
\item $\smash{\left\{ B_{r}^{\ell_{\times}}\right\} _{r}}\subset\L$ satisfy
the following almost-covering property:
\[
\mu\left(\L\setminus B\left(B_{r}^{\ell_{\times}},\varepsilon_{r}\right)\right)\rightarrow0\qquad\text{whenever}\qquad\varepsilon_{r}=\omega\left(\frac{\log^{2}|B_{r}^{\ell_{\times}}|}{|B_{r}^{\ell_{\times}}|^{c/C}}\right),
\]
and in particular, if $c=C$ then the sets $\left\{ B_{r}^{\ell_{\times}}\right\} _{r}$
form an a.o.a.c.\ of $\L$.
\item For any small enough $\varepsilon$, every $g\in\L$ has an $\varepsilon$-approximation
in $B_{r}^{\ell_{\times}}$ for $r=2C\log\tfrac{1}{\varepsilon}+4C\log_{p}\log\tfrac{1}{\varepsilon}$.
\end{enumerate}
\end{prop}

\begin{proof}
\emph{(1)} Recall $d_{r}\leq\left|\Omega_{r}\right|\leq\mathfrak{M}hd_{r}$
and $d_{r}=(p+1)\cdot p^{r-1}$, hence $p^{r}\leq|\Omega_{r}|\leq2\mathfrak{M}hp^{r}$.
Since $B_{r}^{\ell_{CH}}=\bigcup_{i=0}^{r}\Omega_{r}$, we get $p^{r}\leq\left|B_{r}^{\ell_{CH}}\right|\leq4\mathfrak{M}hp^{r}$,
and since $\ell_{\times}$ is $\left(c,C,b\right)$-q.i. to $\ell_{CH}$,
we get $B_{(r-b)/C}^{\ell_{CH}}\subseteq B_{r}^{\ell_{\times}}\subseteq B_{(r+b)/c}^{\ell_{CH}}$,
hence $p^{(r-b)/C}\leq\left|B_{r}^{\ell_{\times}}\right|\leq4\mathfrak{M}hp^{(r+b)/c}$.
This in particular implies 
\[
\varepsilon_{r}=\omega\left(\frac{\log^{2}|B_{r}^{\ell_{\times}}|}{|B_{r}^{\ell_{\times}}|^{c/C}}\right)=\omega\left(\frac{((r-b)/C)^{2}}{\left(p^{(r+b)/c}\right)^{c/C}}\right)=\omega\left(\frac{r^{2}}{p^{r/C}}\right)=\omega\left(\frac{\log^{2}|\Omega_{(r-b)/C}|}{|\Omega_{(r-b)/C}|}\right),
\]
so we can use Corollary~\ref{cor:Clozel-a.o.a.c.} applied to $\left\{ \Omega_{(r-b)/C}\right\} _{r}$
and $\varepsilon_{r}$ we obtain (via $\Omega_{(r-b)/C}\subseteq B_{r}^{\ell_{\times}}$)
\[
\mu\left(\L\setminus B\left(B_{r}^{\ell_{\times}},\varepsilon{}_{r}\right)\right)\leq\mu\left(\L\setminus B\left(\Omega_{(r-b)/C},\varepsilon{}_{r}\right)\right)\overset{{\scriptscriptstyle r\rightarrow\infty}}{\longrightarrow}0.
\]
\emph{(2)} From Proposition~\ref{prop:CH-covering} we know that
$B\left(\Omega_{r'},\varepsilon\right)=\L$ for $\varepsilon$ small
enough and $r'=2\log_{p}\tfrac{1}{\varepsilon}+3\log_{p}\log\frac{1}{\varepsilon}$.
For small enough $\varepsilon$ we also have $C\log_{p}\log\tfrac{1}{\varepsilon}>b$,
so that $r\geq Cr'+b$, and the claim then follow from $\Omega_{r'}\subseteq B_{r'}^{\ell_{CH}}\subseteq B_{Cr'+b}^{\ell}\subseteq B_{r}^{\ell}$.
\end{proof}
One case where $c=C$ indeed occurs is the level sets of the level
map (\ref{eq:l}): 
\begin{equation}
B_{r}^{\ell}=\left\{ \gamma\in\Gamma\,\middle|\,\ell\left(\gamma\right)\leq r\right\} .\label{eq:level-sets}
\end{equation}

\begin{lem}
\label{lem:level-CH} For any $\gamma\in\Gamma$ we have $\left|\ell\left(\gamma\right)-\ell_{CH}\left(\gamma\right)\right|\leq4h$,
i.e. $\ell$ is $(1,1,4h)$-q.i.~to $\ell_{CH}$.
\end{lem}

\begin{proof}
Let $r=\ell_{CH}\left(\gamma\right)$. Then $\gamma\in\Omega_{r}$,
hence $r=\mathrm{dist}\left(\gamma v_{j},v_{i}\right)$, for some
$0\leq i,j<h$. By Proposition \ref{prop:level_dist}, $\ell\left(\gamma\right)=\mathrm{dist}\left(\gamma v_{0},v_{0}\right)$.
Note that $v_{0},v_{i},v_{j}$ belong to the same fundamental domain
of $\Gamma\backslash\mathcal{T}$, which is itself a connected bipartite
graph with $h$ left vertices. It follows that $\mathrm{dist}\left(v_{0},v_{i}\right)\leq2h$
and $\mathrm{dist}\left(\gamma v_{0},\gamma v_{j}\right)\leq2h$,
so that by the triangle inequality we get 
\[
\left|\ell\left(\gamma\right)-\ell_{CH}\left(\gamma\right)\right|=\left|\mathrm{dist}\left(\gamma v_{0},v_{0}\right)-\mathrm{dist}\left(\gamma v_{j},v_{i}\right)\right|\leq\mathrm{dist}\left(\gamma v_{0},\gamma v_{j}\right)+\mathrm{dist}\left(v_{0},v_{i}\right)\leq4h.\qedhere
\]
\end{proof}
From this Lemma together with Proposition~\ref{prop:a.o.a.c-a.c.}
we obtain:
\begin{cor}
\label{cor:level-a.o.a.c.} The sequence $\left\{ B_{r}^{\ell}\right\} _{r}$
of level sets in $\Gamma$ forms an a.o.a.c.\ of $\L$.\hfill{}$\qedsymbol$
\end{cor}

\subsection{Clifford+D Silver gates\protect\label{sec:Silver-gates}}

We now specialize again to the case of the C+D Clifford gates (Definition
\ref{def:C+T-and-C+D}), so that $\Gamma$ is the $3$-arithmetic
lattice in the projective unitary group associated with $\mathbb{Q}\left(\zeta_{9}\right)$
and the standard Hermitian form. The gates $\left\{ H\right\} \cup\mathcal{D}$
generate the group $\Gamma$ (Theorem \ref{thm:main}), and we are
interested in the covering efficiency of balls with respect to the
word/circuit length function:
\[
\ell_{w}\left(\gamma\right):=\min\left\{ r\,\middle|\,\gamma\in\left(\left\{ H\right\} \cup\mathcal{D}\right)^{r}\right\} .
\]
For the analysis we first introduce another length function on $\Gamma$.
Recall that $\Gamma=C_{0}*_{C_{D}}C_{3}$ by Corollary \ref{cor:Bass-Serre}.
The element $S=\diag\left(1,\zeta_{3},1\right)$ (see Definition \ref{def:C+T-and-C+D})
rotates the three $v_{0}$-clans descending from $v_{3}$, so that
$T_{3}:=\left\{ 1,H,HS,HS^{2}\right\} $ is a transversal for the
right cosets $C_{D}\backslash C_{3}$. The set $T_{0}=\left\{ D_{a,b}:=\diag\left((-\xi)^{a},1,(-\xi)^{b}\right)\,|\,0\leq a,b\leq5\right\} $
forms a transversal for $C_{D}\backslash C_{0}$ (and $D_{0,6},D_{6,0}\in C_{D}$),
and we denote $T'_{j}=T\backslash\{1\}$ for $j=0,3$. By the normal
form theorem of Bass-Serre theory, Corollary \ref{cor:Bass-Serre}
implies that every $\gamma\in\Gamma$ has a unique representation
as
\begin{equation}
\gamma=c_{0}\ldots c_{r}\qquad\left({c_{0}\in C_{0}\atop \forall j\geq1\colon\ c_{2j-1}\in T'_{3},c_{2j}\in T'_{0}}\right).\label{eq:BS-normal}
\end{equation}
We call the $r$ which appears in this representation the \emph{Bass-Serre
length }of $\gamma$, and denote it by $\ell_{BS}(\gamma)$. 
\begin{lem}
\label{lem:BS-wordlen}
\begin{enumerate}
\item For any $\gamma\in\Gamma$ we have 
\[
\ell_{BS}(\gamma)-1\leq\ell_{w}(\gamma)\leq\ell_{BS}(\gamma)+7.
\]
\item The size of $r$-balls in both the word length and Bass-Serre metrics
is $\Theta(\sqrt{105}^{\,r})$.
\end{enumerate}
\end{lem}

\begin{proof}
\emph{(1)} In a shortest presentation of $\gamma$ as a word in $\left\{ H\right\} \cup\mathcal{D}$,
there can be no two consecutive elements from $\mathcal{D}$. In addition,
$H^{2}\in S_{3}\leq N_{\Gamma}(\mathcal{D})$, so that any appearance
of $H^{2}$ can be moved to the beginning of the word. Thus we can
assume that $\gamma=H^{\alpha}d_{1}Hd_{2}H\ldots d_{t}H^{\beta}$
with $0\leq\alpha\leq3$ and $0\leq\beta\leq1$, and $\ell_{w}(\gamma)=2t-1+\alpha+\beta$.
Since $\left\langle H\right\rangle \leq C_{3}$ and $\mathcal{D}\leq C_{0}$,
this gives a presentation of $\gamma$ as an element in $C_{3}C_{0}C_{3}C_{0}\ldots$
of length $2t-1+\beta$. Such a word can then be brought to its Bass-Serre
normal form (\ref{eq:BS-normal}); this process does not lengthen
the word, but we might need an extra letter in order to begin in $C_{0}$
(see \cite{serre1980trees}), so that $\ell_{BS}(\gamma)\leq2t+\beta\leq\ell_{w}(\gamma)+1$.
On the other, let $\gamma=c_{0}\ldots c_{r}$ be as in (\ref{eq:BS-normal}).
Then $c_{0}\in C_{0}=S_{3}\mathcal{D}$, and (\ref{eq:M3-fromCD})
(note $\left\langle W\right\rangle \leq\mathcal{D}$) implies that
$c_{0}\in(\{H\}\cup D)^{6}$. For every $j\geq1$, we have $c_{2j-1}c_{2j}\in HS^{m}\mathcal{D}=H\mathcal{D}$
for some $0\leq m\leq2$, and in total we obtain that $\ell_{w}(\gamma)\leq6+2\left\lceil r/2\right\rceil \leq r+7=\ell_{BS}(\gamma)+7$.

\emph{(2)} By the uniqueness of Bass-Serre normal form, for every
$r\geq0$ there are $|C_{0}||T_{3}'|^{\left\lceil r/2\right\rceil }\left|T'_{0}\right|^{\left\lfloor r/2\right\rfloor }=\Theta(\sqrt{105}^{\,r})$
words of Bass-Serre length $r$. This implies $\left|B_{r}^{\ell_{BS}}\right|=\Theta(\sqrt{105}^{\,r})$
as well, and $\left|B_{r}^{\ell_{w}}\right|=\Theta\left(\left|B_{r}^{\ell_{BS}}\right|\right)$
follows from (1).
\end{proof}

Next, we compare the Bass-Serre/word length with the Clozel-Hecke
length.
\begin{lem}
\label{lem:level-H} For any $\gamma\in\Gamma$ we have 
\[
\tfrac{1}{3}\ell_{CH}(\gamma)-4\leq\ell_{w}(\gamma)\leq\ell_{CH}(\gamma)+16.
\]
\end{lem}

\begin{proof}
Note that for any $HS^{m}\in T_{3}'$ and $\gamma\in\Gamma$ we have
\[
\ell\left(\gamma HS^{m}\right)=\dist(\gamma HS^{m}v_{0},v_{0})=\dist(\gamma Hv_{0},v_{0})\leq\dist(\gamma Hv_{0},\gamma v_{0})+\dist(\gamma v_{0},v_{0})=6+\ell\left(\gamma\right),
\]
and for any $c\in C_{0}$ (and in particular for $c\in T_{0}'$) we
have similarly $\ell\left(\gamma c\right)=\ell\left(\gamma\right)$.
By induction, it follows from (\ref{eq:BS-normal}) that $\ell\left(\gamma\right)\leq3\left(\ell_{BS}(\gamma)+1\right)$.
We note that $\overline{G}$ has class number $h=2$, as the vertices
$v_{0},v_{2}$ in (\ref{eq:D-domain}) are representatives for the
orbits of $\Gamma$ on $L_{\mathcal{T}}$, and we obtain by Lemmas
\ref{lem:level-CH} and \ref{lem:BS-wordlen} 
\[
\ell_{CH}(\gamma)\leq\ell(\gamma)+4h\leq3\ell_{BS}(\gamma)+11\leq3\ell_{w}(\gamma)+14.
\]
On the other hand, by the proof of Theorem \ref{thm:main}, for some
$r\leq\ell(\gamma)/2$ there exist $c_{1},c_{2},\ldots,c_{r},c_{r+1}\in C_{0}$
such that $\gamma=c_{r+1}H^{-1}c_{r}^{-1}H^{-1}c_{r-1}^{-1}\ldots H^{-1}c_{1}^{-1}\in C_{0}\left(C_{3}C_{0}\right)^{r}$.
Bringing this to the Bass-Serre normal form of $C_{0}*_{C_{D}}C_{3}$
(i.e.\ (\ref{eq:BS-normal})) can only shorten its length, so that
\[
\ell_{w}(\gamma)\leq\ell_{BS}(\gamma)+7\leq2r+8\leq\ell(\gamma)+8\leq\ell_{CH}(\gamma)+16.\qedhere
\]
\end{proof}
\begin{cor}
\label{cor:H-a.o.a.c.} The sequence $\left\{ B_{r}^{\ell_{w}}\right\} _{r}$
of words of growing lengths in H+D satisfy the following almost-covering
property: 
\[
\mu\left(PU(3)^{3}\setminus B\left(B_{r}^{\ell_{w}},\varepsilon_{r}\right)\right)\rightarrow0,\qquad\mbox{whenever}\qquad\varepsilon_{r}=\omega\left(\frac{\log^{2}|B_{r}^{\ell_{w}}|}{|B_{r}^{\ell_{w}}|^{1/3}}\right).
\]
\end{cor}

\begin{proof}
Follows from Proposition~\ref{prop:a.o.a.c-a.c.} combined with Lemma~\ref{lem:level-H}.
\end{proof}
\begin{rem}
\label{rem:GoldenGates}If the balls w.r.t.\ to the $H$-count had
satisfied the a.o.a.c.\ property, then together with Theorem \ref{thm:main}
which gives a navigation algorithm, and the approximation algorithm
due to Ross--Selinger algorithm \cite{ross2015optimal} (see \cite[Rem.\ 2.9]{Evra2018RamanujancomplexesGolden}),
we would have that the C+D form a super golden gate set for $PU(3)$,
in the terminology of \cite{Parzanchevski2018SuperGoldenGates} and
\cite{Evra2018RamanujancomplexesGolden}, as $H$ and the elements
in $\mathcal{D}$ are all of finite order.
\end{rem}

We also get the following covering result in the form of the Solovay-Kitaev
Theorem, with an optimal exponent (see the introduction).
\begin{thm}
\label{thm:H-covering}For any $K>\log_{3}(105)$, for any small enough
$\varepsilon$ every $g\in PU(3)^{3}$ has an $\varepsilon$-approximation
by a word in the C+D gate set of length $<K\log_{\rho(H\cup\{\mathcal{D}\})}\left(\tfrac{1}{\varepsilon}\right)$.
\end{thm}

\begin{proof}
From Lemma \ref{lem:BS-wordlen} we have $\rho(H\cup\{\mathcal{D}\})=\sqrt{105}$,
and for $\varepsilon$ small enough we have $K\log_{\sqrt{105}}\tfrac{1}{\varepsilon}\geq2\log_{3}\tfrac{1}{\varepsilon}+4\log_{3}\log\tfrac{1}{\varepsilon}$,
so this follows from Proposition \ref{prop:a.o.a.c-a.c.}(2) combined
with Lemma~\ref{lem:level-H}. 
\end{proof}
\bibliographystyle{plain}
\bibliography{mybib}

\begin{flushleft}
\noun{Einstein Institute of Mathematics, Hebrew University of Jerusalem,
Israel.}\texttt{}~\\
Email: \texttt{shai.evra@mail.huji.ac.il, ori.parzan@mail.huji.ac.il.}
\par\end{flushleft}
\end{document}